\documentclass[a4paper,UKenglish,cleveref,thm-restate]{lipics-v2021}
\usepackage{arydshln}
\usepackage{bm}
\usepackage{float}
\usepackage{environ}

\nolinenumbers

\newcommand{\Q}{\mathbb Q}
\newcommand{\R}{\mathbb R}
\newcommand{\Z}{\mathbb Z}
\newcommand{\veczero}{\mathbf0}
\newcommand{\matzero}{\bm O}
\newcommand{\vecone}{\mathbf1}
\newcommand{\vecinfty}{\bm\infty}
\renewcommand{\O}{\mathcal O}
\newcommand{\G}{\mathcal G}
\DeclareMathOperator{\conv}{conv}
\DeclareMathOperator{\ocp}{ocp}

\renewcommand{\pod}[1]{\allowbreak\if@display\mkern18mu\else\mkern8mu\fi(#1)}

\newlength{\vdotsheight}
\makeatletter
\newenvironment{cdisplaymath}{\@fleqnfalse\begin{displaymath}}{\end{displaymath}}
\makeatother

\crefname{lemma}{Lemma}{Lemmata}
\crefname{lemma}{Lemma}{Lemmata}
\crefname{equation}{}{}
\Crefname{equation}{}{}
\newcommand{\crefilp}[1]{ILP~\cref{#1}}

\title{Parameterized Algorithms for Matching Integer Programs with Additional Rows and Columns}

\titlerunning{Parameterized Algorithms for Matching IPs with Additional Rows and Columns}

\author{Alexandra Lassota}{Eindhoven University of Technology, Netherlands}{a.a.lassota@tue.nl}{https://orcid.org/0000-0001-6215-066X}{}

\author{Koen Ligthart}{Eindhoven University of Technology, Netherlands}{k.m.ligthart@tue.nl}{https://orcid.org/0009-0004-6823-5225}{}

\authorrunning{A. Lassota and K. Ligthart}

\Copyright{Alexandra Lassota and Koen Ligthart} 



\relatedversion{An early version of this work appeared in ICALP 2025 \url{https://doi.org/10.4230/LIPICS.ICALP.2025.112}}

\hideLIPIcs

\begin{document}

\maketitle

\begin{abstract}
We study integer linear programs (ILP) of the form $\min\{c^\top x\ \vert\ Ax=b,l\le x\le u,x\in\Z^n\}$ and analyze their parameterized complexity with respect to their distance to the generalized matching problem, following the well-established approach of capturing the hardness of a problem by the distance to triviality. The generalized matching problem is an ILP where each column of the constraint matrix has a $1$-norm of at most $2$. It captures several well-known polynomial time solvable problems such as matching and flow problems. We parameterize by the size of variable and constraint backdoors, which measure the least number of columns or rows that must be deleted to obtain a generalized matching ILP. This extends generalized matching problems by allowing a parameterized number of additional arbitrary variables and constraints, yielding a novel parameter.

We present the following results: (i) a fixed-parameter tractable (FPT) algorithm for ILPs parameterized by the size $p$ of a minimum variable backdoor to generalized matching; (ii) a randomized slice-wise polynomial (XP) time algorithm for ILPs parameterized by the size $p+h$ of a mixed variable plus constraint backdoor to generalized matching as long as $c$ and $A$ are encoded in unary; (iii) we complement (ii) by proving that solving an ILP is W[1]-hard when parameterized by the size of a minimum constraint backdoor $h$ even when $c,A,b,l,u$ have coefficients of constant size. To obtain (i), we prove a variant of lattice-convexity of the degree sequences of weighted $b$-matchings, which we study in the light of SBO jump M-convex functions. This allows us to model the matching part as a polyhedral constraint on the integer backdoor variables. The resulting ILP is solved in FPT time using an integer programming algorithm. For (ii), the randomized XP time algorithm is obtained by pseudo-polynomially reducing the problem to the exact matching problem. To prevent an exponential blowup in terms of the encoding length of $b$, we bound the proximity of the ILP through a subdeterminant based circuit bound. The hardness result (iii) is obtained through a parameterized reduction from ILP with $h$ constraints and coefficients encoded in unary.
\end{abstract}

\section{Introduction}\label{sec:Intro}

We study integer linear programs (ILPs) and analyze their parameterized complexity with respect to their distance to the generalized matching problem. In general, an ILP is of the form
\begin{equation}
    \min\bigl\{c^\top x\bigm\vert Ax=b,l\le x\le u,x\in\Z^n\bigr\},
    \label{ilp:general}\tag{G}
\end{equation}
where $l\in(\Z\cup\{-\infty\})^n,u\in(\Z\cup\{\infty\})^n,c\in\Z^n,b\in\Z^m$ and $A\in\Z^{m\times n}$. The underlying integer linear programming problem is to either decide that the system is infeasible, there exists an optimal feasible solution, or it is unbounded and provide a solution with an unbounded direction of improvement.

Integer linear programming is a powerful language that has been applied to many key problems such as graph problems \cite{DBLP:conf/isaac/FellowsLMRS08,DBLP:journals/dam/FialaGKKK18}, scheduling and bin packing \cite{DBLP:journals/jacm/GoemansR20,DBLP:journals/mp/JansenKMR22}, multichoice optimization \cite{ermolieva2023connections} and computational social choice \cite{bartholdi1989voting,DBLP:journals/teco/KnopKM20}, among others. Unfortunately, solving ILPs is, in general, NP-hard. 

However, most ILP formulations of problems are naturally well-structured. See e.g.~\cite{DBLP:conf/isaac/FellowsLMRS08,DBLP:journals/jacm/GoemansR20,DBLP:journals/algorithmica/GrammNR03,DBLP:journals/mp/JansenKMR22,DBLP:journals/teco/KnopKM20} and the references therein. Motivated by this insight and their daunting general hardness, classes of ILPs with highly-structured constraint matrices and parameterizations have been intensively and successfully studied to obtain polynomial and FPT time algorithms~\cite{DBLP:conf/isaac/FellowsLMRS08,DBLP:journals/jacm/GoemansR20,DBLP:journals/algorithmica/GrammNR03,DBLP:journals/mp/JansenKMR22,DBLP:journals/teco/KnopKM20}. Arguably most famous is Lenstra's 1983 algorithm who presented the first FPT time algorithm for constraint matrices with few columns~\cite{DBLP:journals/mor/Lenstra83}. The body of literature for the over three-decades-long research and the many structures, parameters, and applications are too vast to cover here, we thus refer to~\cite{GavenciakKK22} for an overview.

Central to this paper is the polynomial time solvable class of the generalized matching problem, which is the ILP problem restricted to coefficient matrices $A$ satisfying $\|A\|_1\le2$, i.e., all columns have $1$-norm at most $2$.

\begin{theorem}[Theorem 36.1 in~\cite{schrijver2003combinatorial}]
    When $\|A\|_1\le2$, \crefilp{ilp:general} can be solved in strongly polynomial time.
    \label{thm:generalized-matching-in-p}
\end{theorem}

The generalized matching problem captures a variety of well-known problems in P such as minimum cost flow, minimum cost ($b$-)matching and certain graph factor problems~\cite{schrijver2003combinatorial}. This connection becomes apparent by interpreting the variables as values assigned to edges of a graph whose vertices correspond to the constraints of the ILP\footnote{Note that, due to the conventions in ILP, the vertices of the graph are the numbers $[m]:=\{1,\dots,m\}$ and the edges are identified with the numbers $[n]$.}. In this way, a constraint matrix $A$ with $\|A\|_1\le2$ can be interpreted as an incidence matrix of an associated bidirected graph $G(A)$ where all endpoints of edges are given signs and where degenerate self-loops and half-edges may be present~\cite{DBLP:conf/aussois/EdmondsJ01}. Here, each column or row of $A$ corresponds to an edge or vertex of $G(A)$ respectively and an entry $A_{ij}$ is nonzero if and only if vertex $i$ is incident to edge $j$ in $G(A)$. An example of a graphic interpretation of an instance of the generalized matching problem is shown in \cref{fig:example-generalized-matching-instance}.
\begin{figure}
    \begin{subfigure}{0.45\textwidth}
        \begin{cdisplaymath}
            \begin{pmatrix}
                1&0&1&0&0&0\\
                1&1&0&-1&0&0\\
                0&-1&0&0&2&0\\
                0&0&-1&-1&0&1
            \end{pmatrix}
            x=
            \begin{pmatrix}
                3\\
                1\\
                5\\
                0
            \end{pmatrix}
        \end{cdisplaymath}
        \caption{Example constraints $Mx=b$ with $\|M\|_1=2$ of a generalized matching instance.}
        \label{fig:example-generalized-matching-instance-system}
    \end{subfigure}
    \hspace{0.05\textwidth}
    \begin{subfigure}{0.45\textwidth}
        \centering
        \includegraphics[width=0.8\textwidth]{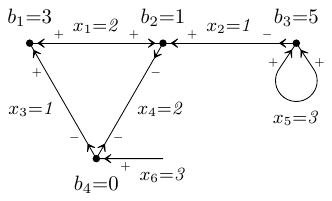}
        \caption{Graphic interpretation of a solution to the system of \cref{fig:example-generalized-matching-instance-system}.}
    \end{subfigure}
    \caption{An example instance of the generalized matching problem and a corresponding solution, disregarding potential variable bounds.}
    \label{fig:example-generalized-matching-instance}
\end{figure}
For instance, a problem such as the minimum cost perfect matching problem on a simple graph $G=(V,E)$ can be cast as a generalized matching ILP. In this scenario, $A\in\{0,1\}^{V\times E}$ is the incidence matrix of $G$ and $l=\veczero,u=\vecone,b=\vecone$. Here, $\veczero$ and $\vecone$ denote the all-zero and all-one vector respectively.

This paper studies the degree to which we can extend the generalized matching problem while maintaining tractability. For this purpose, we study constraint matrices $A$ which are similar to the constraint matrix of a generalized matching problem. Such matrix $A$ consists of a small corner block $C\in\Z^{h\times p}$, wide block $W\in\Z^{h\times n}$, tall block $T\in\Z^{m\times p}$ and matching-like block $M\in\Z^{m\times n}$ with $\|M\|_1\le2$ and is of the form
\[
    A=\begin{pmatrix}
        C&W\\
        T&M
    \end{pmatrix}
\]
with associated ILP
\begin{equation}
    \begin{aligned}
        \min\bigl\{a^\top y+c^\top x\bigm\vert&\ Cy+Wx=d,Ty+Mx=b,\\
        &\ e\le y\le g,l\le x\le u,(y,x)\in\Z^{p+n}\bigr\}.
    \end{aligned}
    \label{ilp:mixed}\tag{M}
\end{equation}
We separately consider the special cases where either $h=0$ or $p=0$, yielding the ILPs
\begin{equation}
    \min\bigl\{a^\top y+c^\top x\bigm\vert Ty+Mx=b,e\le y\le g,l\le x\le u,(y,x)\in\Z^{p+n}\bigr\},
    \label{ilp:tall}\tag{T}
\end{equation}
which describes a generalized matching ILP admitting additional variables, and
\begin{equation}
    \min\bigl\{c^\top x\bigm\vert Wx=d,Mx=b,l\le x\le u,x\in\Z^n\bigr\},
    \label{ilp:wide}\tag{W}
\end{equation}
which describes a generalized matching ILP admitting additional constraints. \crefilp{ilp:mixed} represents a generalization of well-known flow and matching problems. See~\cite{DBLP:journals/networks/BalasP83} for an application of bipartite matching with additional variables to a scheduling problem.

In the context of the fixed-parameter tractability of ILPs, graphs and corresponding structural parameters associated with the coefficient matrix of an ILP are usually studied. However, in our case, the coefficient matrix of the perfect matching problem may have arbitrary associated primal or dual graphs, unlike the case for two-stage stochastic and $n$-fold ILPs, which have associated coefficient matrix graphs with limited tree-depth. See~\cite{eisenbrand2022algorithmictheoryintegerprogramming} for an overview. In addition, the incidence matrix of an undirected graph may have have arbitrarily large subdeterminants, which makes it unsuited for methods such as the one discussed in~\cite{DBLP:conf/focs/FioriniJWY21}.

The goal of this paper is to study the complexity of the above integer linear programs \labelcref{ilp:mixed,ilp:tall,ilp:wide}, i.e., ILPs where the constraint matrices are \emph{nearly} generalized matching constraint matrices. For this purpose, we parameterize by the number of variables $p$ and constraints $h$ that must be deleted from \crefilp{ilp:general} to obtain a generalized matching problem. Such parameter choice follows the classical approach of studying parameters measuring the \emph{distance to triviality}, also called \emph{(deletion) backdoors to triviality}--a concept that was first proposed by Niedermeier~\cite{DBLP:books/ox/Niedermeier06}. Roughly speaking, this approach introduces a parameter that measures the distance of the given instance from an instance that is solvable in polynomial time. This approach was already used for many different problems such as clique, set cover, power dominating set, longest common subsequence, packing problems, and the satisfiability problem~\cite{DBLP:books/ox/Niedermeier06,DBLP:conf/mfcs/BannachBMMLRS20,DBLP:journals/jacm/GoemansR20,DBLP:conf/birthday/GaspersS12}. 

For some problems, such as satisfiability, having obtained a variable backdoor immediately leads to an FPT algorithm parameterized by the size of the backdoor. For ILP, however, this is not as straightforward as variable domains may be arbitrarily large, which invalidates a brute-force approach to obtain FPT results in terms of variable backdoor size. In fact, solving ILPs in FPT time parameterized by the number of variables is already highly nontrivial~\cite{DBLP:journals/mor/Lenstra83,DBLP:conf/focs/ReisR23}.

Despite such difficulties, a significant number of results have been obtained in the context of backdoors in ILPs. The complexity of ILPs which become totally unimodular after removing a constant number of columns and rows is actively being researched~\cite{DBLP:conf/soda/AprileFJKSWY25,DBLP:journals/mor/NageleSZ24}. In addition, a large line of research has established efficient algorithms for solving ILPs given a backdoor to a remaining ILP which consists of many isolated ILPs with small coefficients~\cite{DBLP:conf/soda/CslovjecsekKLPP24,DBLP:journals/ai/DvorakEGKO21,eisenbrand2022algorithmictheoryintegerprogramming}. A two-stage stochastic ILP, which is an important special case of such block-structured ILPs, may be viewed as an ILP of the form shown in \labelcref{ilp:tall} where $M$ is a block-diagonal matrix with small nonzero blocks and small coefficients. To relate these results with backdoors and backdoor identification, Dvořák et al.~\cite{DBLP:journals/ai/DvorakEGKO21} define fracture numbers, which describe the number of variable and/or constraint removals needed to decouple the ILP into small blocks. They give an overview of the parameterized complexity of ILPs parameterized by various backdoor sizes and coefficient sizes of the constraint matrix. Their work and the recent result of Cslovjecsek et al.~\cite{DBLP:conf/soda/CslovjecsekKLPP24} shows that the complexity of ILPs parameterized by backdoor size is subtle, as ILPs may already become NP-hard when there are only 3 complicating global variables connecting constant dimension, otherwise independent ILPs with unary coefficient size~\cite{DBLP:journals/ai/DvorakEGKO21}. On the other hand, when additionally parameterizing by the coefficient size of the small blocks, the problem becomes FPT~\cite{DBLP:conf/soda/CslovjecsekKLPP24}. With this paper, we aim to reveal the parameterized complexity with respect to backdoor sizes to another class of well-known efficiently solvable ILPs, that of the generalized matching problem.

We show that ILP parameterized by $p$, the number of variables that, upon deletion, give a generalized matching problem, is FPT by designing a corresponding algorithm. Further, we give a randomized, XP time algorithm with respect to $p+h$, where $p$ and $h$ are the number of variables and constraints, respectively, of which the deletion yields a generalized matching problem, and prove the corresponding W[1]-hardness result for this case.

We note that studying the complexity of ILP with respect to $p$ is also motivated by the hardness of the general factor problem with a gap size of $2$~\cite{lovasz1972factorization}. Lovász's reduction from edge coloring on cubic graphs shows that, even when $T$ consists solely of columns with only one nonzero coefficient with value $3$, ILP remains NP-hard. This, together with the tractability of generalized matching and our FPT result, shows that if one were to judge the complexity of an ILP by independently and individually labeling each column as ``hard'' or ``easy'', the easy columns are precisely those with $1$-norm at most $2$.

Finally, given an arbitrary constraint matrix $A$ of \crefilp{ilp:general}, a permutation of the columns or rows to obtain the form \crefilp{ilp:tall} or \labelcref{ilp:mixed} with minimum $p$ or $p+h$ respectively can be obtained efficiently\footnote{Through the use of elementary row operations, a system may be modified to a row-equivalent system admitting smaller backdoors. See~\cite{DBLP:journals/mp/BrianskiKKPS24} for such study on block-structured ILP. We leave this problem open in the context of the parameterizations discussed in this paper.}. That is, one can identify a minimum cardinality variable or mixed variable and constraint deletion backdoor to the generalized matching ILP. To obtain the form \crefilp{ilp:tall}, one may greedily move all columns with $1$-norm greater than $2$ to the left. The problem of minimizing $p+h$ in \crefilp{ilp:mixed} contains the NP-hard $3$-uniform hitting set problem, cf.~\cite{DBLP:conf/birthday/GaspersS12}, but is FPT as it can be solved with a bounded depth search tree algorithm branching on columns with $1$-norm at least $3$. Such column must either be assigned to be a column of $T$, or in an arbitrary subset of at most $3$ entries with absolute sum of at least $3$, at least one of these entries must be part of a row in $W$. For this reason, we will assume that the given ILPs are of the form \labelcref{ilp:tall} or \labelcref{ilp:mixed} throughout the rest of the paper.

\subsection{Contributions}

In this paper, we show that solving integer linear programs is fixed-parameter tractable parameterized by the number of columns of $A$ with $1$-norm greater than $2$.

\begin{restatable}{theorem}{thmtallfpt}
    \crefilp{ilp:tall} is FPT parameterized by $p$.
    \label{thm:tall-fpt}
\end{restatable}

This reveals an entirely new class of ILPs that is FPT and adds to the story of parameterized ILPs and distance to triviality paradigm. To obtain this algorithm, we use a remainder guessing strategy introduced by Cslovjecsek et al.~\cite{DBLP:conf/soda/CslovjecsekKLPP24}. Central in this approach is modeling non-backdoor variables as polyhedral constraints on the backdoor variables. In~\cite{DBLP:conf/soda/CslovjecsekKLPP24}, these polyhedral constraints may be exponentially complex, whereas we exploit the structure of matching problems and employ an efficient description of a global polyhedral constraint. To accomplish this, we study the convexity of degree sequences for which subgraphs with matching degrees exist. These systems have been studied previously, see~\cite{DBLP:journals/jgt/AnsteeN99}, and are a well-known example of discrete systems known as jump systems~\cite{DBLP:journals/siamdm/BouchetC95,DBLP:journals/siamdm/Murota06,DBLP:journals/ieicet/MurotaT06,DBLP:conf/ipco/Kobayashi23}. For our application, we prove a new result involving the lattice-convexity of a particular class of jump M-convex functions on the shifted lattice $2\Z^m+r$. See \cref{sec:overview} for a technical overview of this algorithm. 

For the mixed variable and constraint backdoors, we show that matching-like ILPs are solvable in polynomial time with a randomized algorithm for a fixed number of additional complicating variables and constraints if the constraint matrix and the objective are encoded in unary. As \crefilp{ilp:mixed} can encode the NP-hard subset sum problem for $p=0,h=1$, we need to assume that the coefficients of the constraint matrix are small. That is, we assume that they are bounded by $\Delta$ in absolute value. The given randomized XP time algorithm unifies the known tractability in terms of membership in RP of multiple classes of constrained problems where the edges of a graph correspond to variables.

\begin{restatable}{theorem}{thmmixedxp}
    \crefilp{ilp:mixed} can be solved with a randomized XP time algorithm in terms of $p$ and $h$. More specifically, it can be solved in time \[
        \|c\|_\infty^{\O(1)}\cdot n^{\O(p+h)}\cdot(\Delta(m+h))^{\O((p+h)^2)}.
    \]
    \label{thm:mixed-xp}
\end{restatable}

In this paper, the running time is measured in the number of arithmetic operations on numbers with encoding length polynomial in the encoding length of the instance.

Camerini, Galbiati and Maffioli~\cite{DBLP:journals/jal/CameriniGM92} already observed in 1992 that one can solve constrained flow, circulation and matching problems in randomized, pseudo-polynomial time. \cref{thm:mixed-xp} differs in that it reveals the tractability of a generalized class of problems by removing the exponential dependency on the encoding length of $b$ that would appear in a purely pseudo-polynomial algorithm and admitting a constant number of additional complicating variables.

The randomized XP time algorithm is obtained by polynomially bounding the proximity for fixed $p+h$, which measures the distance between optimal IP solutions from a given optimal LP solution. This implies that we may restrict the variable domains and the right-hand side vector $b$ to a polynomial range and enables a consecutive pseudo-polynomial reduction to the exact matching problem. To complement the proximity upper bound, we additionally provide a proximity lower bound of which the exponent of $m$ scales with $p+h$, which also motivates the use of the different technique for solving \crefilp{ilp:tall}. See \cref{sec:overview-mixed} for the technical overview.

We provide a W[1]-hardness result in terms of the number of additional constraints $h$ to match the randomized XP time algorithm in a complexity theoretic sense. The hardness persists even for $0/1$-constrained perfect matching on a restricted class of graphs. Furthermore, \cref{thm:wide-w1-hard} shows that additionally parameterizing with $\Delta$ is unlikely to result in an FPT algorithm.

\begin{restatable}{theorem}{thmwidewonehard}
    Solving \crefilp{ilp:wide} encoded in unary is W[1]-hard parameterized by $h$. This hardness persists when restricting to $M$ being the incidence matrix of a simple graph which is the disjoint union of even length cycles and $W\in\{0,1\}^{h\times n},d\in\{0,1\}^h,c=\veczero,l=\veczero,u=\vecone,b=\vecone$.
    \label{thm:wide-w1-hard}
\end{restatable}

We like to conclude the introduction with noting that recent methodology used by Eisenbrand and Rothvoss~\cite{eisenbrand2025parameterizedlinearformulationinteger} can be used to derive similar results to \cref{thm:tall-fpt}, which is presented in \cref{sec:milp-approach-for-tall-fpt}. This work was carried out independently from ours.

\section{Technical Overview}\label{sec:overview}
This section gives a technical overview of our results. It aims at presenting the main ideas and new concepts. Therefore, details are omitted. 

We first preprocess an ILP instance in \cref{sec:overview-reduction-to-perfect-b-matching}, after which both algorithms are covered in \cref{sec:overview-few-arbitrary-variables,sec:overview-mixed}. The W[1]-hardness of \cref{thm:wide-w1-hard} is additionally discussed in the latter section.

\subsection[Reduction to perfect b-matching]{Reduction to perfect $b$-matching}
\label{sec:overview-reduction-to-perfect-b-matching}
Before presenting the key components of the algorithms that solve the generalized matching problem with additional variables or constraints, we first discuss a reduction that modifies \crefilp{ilp:mixed} in a favorable way. By doing so, it suffices to give algorithms for restricted instances, see \cref{sec:overview-few-arbitrary-variables,sec:overview-mixed}.

In particular, we show that \crefilp{ilp:mixed} can, in polynomial time, be transformed to an equivalent ILP where the generalized matching submatrix is transformed to the constraint matrix of the perfect $b$-matching problem. Such matrix $M$ is a binary matrix with unique columns that each have precisely two nonzero entries. That is, we can transform the instance so that the associated graph $G(M)$ is a simple graph. In this situation, this graph coincides with the \emph{constraint graph} of the constraint matrix. In terms of ILPs, a perfect $b$-matching $x$ of a simple graph $G(M)$ is a solution to the generalized matching ILP with constraint matrix $M$, nonnegative integral variables $x\in\Z_{\ge0}^n$, and right-hand-side $b$. That is, $x$ assigns a value to every edge such that for every vertex $i$ the sum of $x_j$ over the incident edges $j$ is exactly $b_i$.

The used reduction is an extension of the reduction in~\cite{schrijver2003combinatorial} from generalized matching to $b$-matching. However, we show in \cref{lemma:master-reduction} that we can also keep track of the additional general variables and constraints.  

\begin{restatable}{lemma}{lemmamasterreduction}
    An \crefilp{ilp:mixed} for which the variable bounds are finite, can be reduced in input+output-linear time to an instance of \crefilp{ilp:mixed} that additionally satisfies
    \begin{itemize}
        \item $G(M')$ is simple,
        \item $c'\ge\veczero,W'\in\Z_{\ge0}^{h'\times n'}$,
        \item $e',l'=\veczero,\|g'\|_1=\O(\|g-e\|_1),u'=\vecinfty$,
        \item $n',m'=\O(n+m)$,
        \item $h'=h,p'=p$,
        \item $\|d'\|_1=\O(\|d\|_1+\|A\|_1(\|e\|_1+\|l\|_1))$,
        \item $\|b'\|_1=\O(\|u-l\|_1)$ if $p=0$ and $\|b'\|_1=\O(\|b\|_1+\|g-e\|_1+\|u-l\|_1+\|A\|_1(\|e\|_1+\|l\|_1))$ otherwise,
        \item $a',c',C',T',W'$ contain the same collections of entries as $a,c,C,T,W$ up to sign changes, the insertion of zeros and the insertion of at most one binary row into $T$.
    \end{itemize}
    \label{lemma:master-reduction}
\end{restatable}

Here, parameters with primes are used to denote the parameters of the new instance and $u=\vecinfty$ denotes the absence of upper bounds on $x$. The proof transforms the constraint matrix and variable bounds to the intended form by creating additional constraints and variables. It consists of steps that individually treat columns with negative entries, columns with $1$-norm strictly less than $2$ and finite variable upper bounds. Most of these obstacles are circumvented by conceptually subdividing edges or adding redundant constraints. Due to space restrictions, we postpone the proof to \cref{sec:master-reduction}. 

By virtue of \cref{lemma:master-reduction}, we can now restrict ourselves to finding algorithms for \crefilp{ilp:mixed} which represent perfect $b$-matching problems with additional variables and constraints.

\subsection{Few arbitrary variables}
\label{sec:overview-few-arbitrary-variables}

We now focus on solving \crefilp{ilp:tall}. We first motivate the main idea underlying the algorithm and present the algorithm itself, before we give an outline of the remaining parts. The key idea is to treat the backdoor variables $y$ and the matching variables $x$ separately. Let us first consider the problem of determining the feasibility of \crefilp{ilp:tall}. In essence, we replace the constraint $Ty+Mx=b$ and matching variables $x$ with a polyhedral constraint of the form $b-Ty\in P$ for some polyhedron $P$ that models that there must exist a $x\in\Z_{\ge0}^n$ such that $Mx=b-Ty$, in other words, there must exist a perfect $(b-Ty)$-matching in $G(M)$. To construct such polyhedral constraint on $y$, the first idea may be to use the convex hull of all points $z\in\Z^m$ such that there exists a perfect $z$-matching--unfortunately, a naive execution of this strategy fails as when $G(M)=K_3$, the complete graph on $3$ vertices, we have that $z=(0,0,0)$, $z=(2,2,2)$ and the convex combination $z=(1,1,1)$ are in the hull, but the latter does not admit a perfect $z$-matching. Hence, we attempt to model a discrete set that is not lattice-convex on $\Z^3$ with a convex constraint, which is impossible.

To contrast this, the work by Cslovjecsek et al.~\cite{DBLP:conf/soda/CslovjecsekKLPP24} shows that, even for general matrices $M$, one can work around this issue by restricting $y$ to a fixed remainder $r\in\Z^m$ modulo some large integer $B$. That is, one can construct a polyhedron $P_r$ such that for all $y\in B\Z^m+r$, we have that $\{Mx=b-Ty,x\in\Z_{\ge0}^n\}$ is feasible if and only if $b-Ty\in P_r$. Note that through an affine transformation, this may equivalently be posed as a polyhedral constraint directly on $y$. Their choice of $B$ may grow with the dimensions of $M$, which calls for a more problem specific approach to solve \crefilp{ilp:tall} in FPT time. When $G(M)$ is simple, we show that we may restrict the modulus $B$ to be the constant $2$, which paves the way for an FPT algorithm. In addition, we show that one may optimize over a linear objective function by encoding the objective contribution of the matching part of the ILP in an additional dimension of $P_r$. The required convexity property to make this work translates to a lattice-convexity property, see \cref{lemma:convexity-of-b-matching}, of the function $f_{c,M}$ from \cref{def:f}.

\begin{definition}
    Let $f_{c,M}:\Z^m\to\Z\cup\{\infty\}$ be defined by
    \[
        f_{c,M}(z)=\min\bigl\{c^\top x\bigm\vert Mx=z,x\in\Z_{\ge0}^n\bigr\}.
    \]
    That is, $f_{c,M}(z)$ is the cost of a minimum $c$-cost perfect $z$-matching of $G(M)$ or $\infty$ if $G(M)$ has no perfect $z$-matching.
    \label{def:f}
\end{definition}

\begin{lemma}
    Let $z^{(1)},z^{(2)},\dots,z^{(\ell)}\equiv r\pmod2$ be given and $\lambda^{(1)},\lambda^{(2)},\dots,\lambda^{(\ell)}\ge0$ be real convex multipliers, i.e., $\sum_{k\in[\ell]}\lambda^{(k)}=1$, such that $z:=\sum_{k\in[\ell]}\lambda^{(k)}z^{(k)}\equiv r\pmod2$. Then $f_{c,M}(z)\le\sum_{k\in[\ell]}\lambda^{(k)}f_{c,M}(z^{(k)})$.
    \label{lemma:convexity-of-b-matching}
\end{lemma}

We use modular congruence on vectors to denote component-wise modular congruence. Next, we show in which way \cref{lemma:convexity-of-b-matching} allows us to reformulate \crefilp{ilp:tall}. First, we perform a binary search on the objective to obtain the problem of finding a feasible point of \crefilp{ilp:tall} subject to an additional linear constraint $a^\top y+c^\top x\le\omega^*$. We guess the remainder vector $t$ of $y$ in an optimal solution modulo $2$. It then suffices to search for an integer point $y$ restricted to $y\equiv t\pmod2$ in the box $e\le y\le g$ subject to the constraint that there is a perfect $(b-Ty)$-matching with cost at most $\omega^*-a^\top y$. This is equivalent to finding a point in the set
\[
    \{y\in\Z^p:y\equiv t\pmod2,e\le y\le g,f_{c,M}(b-Ty)\le\omega^*-a^\top y\}.
\]
In particular, as the parity of $b-Ty\equiv r\pmod2$ is constant for $y\equiv t\pmod2$, \cref{lemma:convexity-of-b-matching} shows that there exists a convex set $P_r$ such that the constraint $f_{c,M}(b-Ty)\le\omega^*-a^\top y$ may be replaced with $(\omega^*-a^\top y,b-Ty)\in P_r$. When we obtain a good representation of $P_r$, we use the algorithm by Reis and Rothvoss~\cite{DBLP:conf/focs/ReisR23} to solve the integer program in FPT time. After guessing all $2^p$ parity vectors $t$, an optimal solution to \crefilp{ilp:tall} is found.

We now discuss the missing ingredients needed to make the previously discussed algorithm work. First, we explain how \cref{lemma:convexity-of-b-matching} is derived. To do so, we consider jump systems, where degree sequences of graphs and functions such as $f_{c,M}$ have been studied exhaustively~\cite{DBLP:journals/siamdm/BouchetC95,DBLP:journals/siamdm/Murota06,DBLP:journals/ieicet/MurotaT06,DBLP:conf/ipco/Kobayashi23}. Murota~\cite{DBLP:journals/ieicet/MurotaT06} observed that $f_{c,M}$ describes a jump M-convex function, which is a valuated generalization of jump systems. In fact, recent work on the general factor problem~\cite{DBLP:conf/ipco/Kobayashi23} reveals that a similar jump M-convex function defined for weighted graph factorizations has additional exploitable properties. Informally, small steps, as in \cref{def:2-step-decomposition}, connecting points in the effective domain of such function may be rearranged in any order. For this reason, Kobayashi defines strongly base orderable (SBO) jump systems and a valuated extension. Our function $f_{c,M}$ defined over weighted uncapacitated perfect $b$-matchings can be seen to also satisfy the properties of \cref{def:sbo-jump-m-convex}. As \cref{lemma:convexity-of-b-matching} holds for general SBO jump M-convex functions, we discuss our proof in terms of this abstract property. We use $\sqsubseteq$ to denote the conformal partial order defined by $x\sqsubseteq y$ if and only if $|x_i|\le|y_i|$ and $x_iy_i\ge0$ for all $i$.

\begin{restatable}{definition}{deftwostepdecomposition}
    A $2$-step decomposition of a vector $d\in\Z^m$ is a multiset of steps $p^{(1)},\dots,p^{(\ell)}\in\Z^m$, satisfying $\|p^{(k)}\|_1=2$ and $p^{(k)}\sqsubseteq d$ for all $k\in[\ell]$, such that $d=\sum_{k\in[\ell]}p^{(k)}$.
    \label{def:2-step-decomposition}
\end{restatable}

\begin{restatable}{definition}{defsbojumpmconvex}
    A function $f\colon\Z^m\to\Q\cup\{\infty\}$ is SBO jump M-convex if for every two points $z^{(1)},z^{(2)}$ in the domain $\{z\in\Z^m:f(z)<\infty\}$ there exists a 2-step decomposition $p^{(1)},\dots,p^{(\ell)}$ of $z^{(2)}-z^{(1)}$ and real numbers $g^{(1)},\dots,g^{(\ell)}$ such that
    \begin{itemize}
        \item $f(z^{(2)})=f(z^{(1)})+\sum_{i\in[\ell]}g^{(k)}$,
        \item for all $I\subseteq[\ell]$ it holds that $f(z^{(1)}+\sum_{k\in I}p^{(k)})\le f(z^{(1)})+\sum_{k\in I}g^{(k)}$.\\
    \end{itemize}
    \label{def:sbo-jump-m-convex}
\end{restatable}

The alternating path argument employed by Kobayashi~\cite{DBLP:conf/ipco/Kobayashi23} to show that $\vecone$-capacitated perfect $b$-matchings satisfy \cref{def:sbo-jump-m-convex} can be adapted to show the SBO jump M-convexity of $f_{c,M}$. For this, we use a known gadget construction that relates $b$-matchings to perfect matchings~\cite{schrijver2003combinatorial}. See \cref{sec:tall} for the full proof.

\cref{lemma:convexity-of-b-matching} is derived by first showing that a $2$-step decomposition of a difference $z^{(2)}-z^{(1)}$ can be used to construct a small even step $\sum_{k\in I}p^{(k)}\in\{-2,0,2\}^m$ by combining some of the steps. Such a small even step can be used to bring a pair of points $z^{(2)},z^{(1)}$ closer to each other while remaining on the lattice $2\Z^m+r$ and without increasing the sum of the function values on these points. Such pairwise modifications may then be exhaustively performed to move the points appearing in a convex combination towards the target $z$ as in \cref{lemma:convexity-of-b-matching} and obtain the required result.

The final ingredient needed to implement the FPT algorithm for \crefilp{ilp:tall} is a good characterization of a suitable polyhedron $P_r$ that models the matching part of the ILP. We obtain this by letting $P_r$ be the convex hull of the vectors $(\omega,z)\in\R\times(2\Z^m+r)$ for which $f_{c,M}(z)\le\omega$ and $z$ is in some bounding box. That is, $P_r$ is the epigraph of a natural convex extension of $f_{c,M}$ restricted to a bounding box. Since proximity arguments allow us to assume that $e$ and $g$ are finite, it is possible to bound $z$, see~\cite{DBLP:journals/mp/CookGST86}. We obtain a polynomial time separation oracle for $P_r$ through the ellipsoid method~\cite{DBLP:books/sp/GLS1988} and by noting that $P_r$ can efficiently be optimized over\footnote{We note that it may be possible to derive a combinatorial separation algorithm, but this is not needed to establish the fixed-parameter tractability of \crefilp{ilp:tall}. For a closely related separation problem, see~\cite{DBLP:journals/mor/Zhang03}.}. This follows from the fact that weights associated with a vertex $i$ may be transferred to the edge weights of the incident edges, which shows that optimizing over $(\omega,z)\in P_r$ corresponds to solving a minimum cost parity constrained $b$-matching problem. The parity constraints can straightforwardly be translated into constraints compatible with the polynomially solvable generalized matching problem~\cite{schrijver2003combinatorial,DBLP:journals/mp/EdmondsJ73}.

\subsection{Few arbitrary variables and constraints}
\label{sec:overview-mixed}

The randomized XP time algorithm for solving \crefilp{ilp:mixed} is obtained through a series of reductions. As a key intermediate step, we will employ a pseudo-polynomial reduction in~\cite{schrijver2003combinatorial} that transforms the minimum cost perfect $b$-matching to the minimum cost perfect matching problem on a graph with a total of $\|b\|_1$ vertices. In order to prevent an exponential blow-up in terms of the encoding size of $b$ and to simultaneously eliminate the fixed number $p$ of arbitrary variables, we restrict the search of the optimal solution to a domain with polynomial variable bounds. This is accomplished by bounding the proximity, which limits the maximum distance between LP relaxation solutions and optimal integral solutions. We obtain this by bounding the so-called circuits of the constraint matrix.

\begin{definition}
    A circuit of $A\in\Z^{m\times n}$ is an element $c\in\Z^n$ such that $Ac=\veczero$ with minimal support and coprime entries. The circuit complexity $c_\infty$ of $A$ denotes
    \[
        c_\infty(A)=\max\{\|c\|_\infty\ \vert\ c\text{ is a circuit of }A\}.
    \]
\end{definition}

We adapt the standard subdeterminant bound for $c_\infty(A)$ using Cramer's rule~\cite{onn2010nonlinear}, which on first glance would yield an unsatisfactory $2^{\Omega(n)}$ bound as $G(M)$ may contain many disjoint odd cycles. To deal with this, we show that if an entry of a Cramer's rule solution is sufficiently large, there is a sufficiently large common factor dividing all of the entries simultaneously. This can then be used to show that the normalized circuit is polynomially bounded. To prove \cref{lemma:circuit-ub-mixed}, we Laplace-expand along the arbitrary columns and rows and analyze the structure of the resulting submatrix of $M$ using known properties of the subdeterminants of such matrices established by Grossman et al.~\cite{GROSSMAN1995213}.

\begin{restatable}{lemma}{lemmacircuitubmixed}
    Let
    \[
        A=\begin{pmatrix}C&W\\T&M\end{pmatrix}
    \]
    consist of the matrices $C\in\{-\Delta,-\Delta+1,\dots,\Delta\}^{h\times p},W\in\{-\Delta,-\Delta+1,\dots,\Delta\}^{h\times n},T\in\{-\Delta,-\Delta+1,\dots,\Delta\}^{m\times p}$ and $M\in\Z^{m\times n}$, which satisfies $\|M\|_1\le2$. Let $r$ be the rank of $A$, then $c_\infty(A)\le\O(\Delta r)^{p+h}.$
    \label{lemma:circuit-ub-mixed}
\end{restatable}

\cref{thm:proximity-from-circuits} from~\cite{DBLP:journals/mp/HemmeckeKW14} then immediately yields the polynomial proximity upper bound from \cref{cor:proximity-ub}.

\begin{theorem}[Theorem 3.14 in~\cite{DBLP:journals/mp/HemmeckeKW14}]
    Let $x^*$ be an optimal solution to \crefilp{ilp:general}. If the ILP is feasible, then there exists an optimal ILP solution $x$ with $\|x^*-x\|_\infty\le n\cdot c_\infty(A)$.
    \label{thm:proximity-from-circuits}
\end{theorem}

\begin{corollary}
    Let $(y^*,x^*)$ be an optimal solution to the LP relaxation of \crefilp{ilp:mixed}. If the ILP is feasible, then there exists an optimal ILP solution $(y,x)$ with $\|(y^*,x^*)-(y,x)\|_\infty\le n\cdot\O(\Delta(m+h))^{p+h}$.
    \label{cor:proximity-ub}
\end{corollary}

We prove \cref{lemma:circuit-ub-mixed} in \cref{sec:proximity-bound} and additionally show that this is asymptotically tight for constant $p+h$ and when $m=\Theta(n)$ through an explicit construction which is summarized in \cref{prop:proximity-lb}. We note that a constant multiplicative factor of $1/(p+h)^{p+h+1}$ is hidden in the $\Omega$.

\begin{restatable}{proposition}{propproximitylb}
    For any fixed nonnegative integers $p,h$ and for any positive integer $\Delta$, there exists an infinite family of instances of \crefilp{ilp:mixed}, each given a vertex solution to the LP relaxation, such that the nearest integral solution is at $\infty$-norm distance $\Omega(m\cdot(\Delta m)^{p+h})$ from the relaxation solution. These instances satisfy $m=\Theta(n),l=\veczero,u=\vecinfty$ and use constraint matrix coefficients in $\{0,1,\Delta\}$.
    \label{prop:proximity-lb}
\end{restatable}

Note that \cref{prop:proximity-lb} implies that the circuits and Graver basis elements of the constraint matrix of \crefilp{ilp:mixed} can have an $\infty$-norm of $\Omega((\Delta m)^{p+h})$. This shows that a Graver basis augmentation based approach is unlikely to assist in deriving simpler FPT algorithms to solve \crefilp{ilp:tall} even when $\Delta$ is bounded.

The randomized XP time algorithm which, after establishing \cref{cor:proximity-ub}, boils down to a chain of reductions, is elaborately described in \cref{sec:mixed-algorithm}. We start by computing an optimal LP relaxation solution to the relaxation of \crefilp{ilp:mixed} in polynomial time. Then using \cref{cor:proximity-ub}, we can restrict the search of an optimal integral solution to a domain with $\|(g,u)-(e,l)\|_\infty=n\cdot\O(\Delta(m+h))^{p+h}$. We may then guess the possible values of the $p$ arbitrary variables $y$ in a polynomial number of guesses and obtain an instance of \crefilp{ilp:wide} with polynomially bounded variable domains. \cref{lemma:master-reduction} then shows that this reduces to a constrained perfect $b$-matching problem with polynomially bounded right-hand-side $b$. Now, after applying the pseudo-polynomial reduction to a perfect matching problem~\cite{schrijver2003combinatorial}, which is compatible with additional constraints as noted in~\cite{DBLP:journals/jal/CameriniGM92}, we find that solving \crefilp{ilp:mixed} can be polynomially reduced to finding a minimum cost constrained perfect matching in a graph with polynomially many vertices. As all variables in this resulting problem are binary, we can condense all $h$ constraints into a single constraint by encoding the constraints in base-$B$ for some sufficiently large $B$. This generalizes the reduction steps used in~\cite{DBLP:journals/mp/BergerBGS11,DBLP:journals/jacm/PapadimitriouY82} and only increases the coefficient size of the constraint matrix polynomially for fixed $h$. The resulting constrained perfect matching problem can then be solved with a randomized pseudo-polynomial algorithm such as the one by Mulmuley, Vazirani and Vazirani~\cite{DBLP:journals/combinatorica/MulmuleyVV87}. Combining all steps yields the randomized XP time algorithm from \cref{thm:mixed-xp}.

To the best of our knowledge, the complexity of constrained matching where the objective is encoded in binary is still open. It is worth noting that the reduction chain shows that \crefilp{ilp:mixed} is polynomially equivalent to the exact matching problem, which aims to find a perfect matching in a graph with exactly a given target weight. Finding a deterministic pseudo-polynomial algorithm for this problem has been a central and intensively studied open problem for over four decades.

Finally, to complement the XP time complexity of the algorithm from \cref{thm:mixed-xp}, we reveal the W[1]-hardness of solving \crefilp{ilp:wide} parameterized by $h$ in \cref{thm:wide-w1-hard}. This result is obtained through a parameterized reduction from the ILP problem, which is strongly W[1]-hard when parameterized by the number of constraints as shown by Dvořák et al.~\cite{DBLP:journals/ai/DvorakEGKO21}. We merely rephrase their result and construction in \cref{thm:w1-multicolorclique}.

\begin{restatable}[Adjusted Theorem 22 from~\cite{DBLP:journals/ai/DvorakEGKO21}]{theorem}{thmwonemulticolorclique}
    Determining the feasibility of the system
    \[
        \{Wx=d:x\in\Z_{\ge0}^n\},
    \]
    where $W\in\Z_{\ge0}^{h\times n}$ is encoded in unary and $d\in\{0,1\}^h$, is W[1]-hard parameterized by $h$. This hardness persists when restricting to binary variables $x\in\{0,1\}^n$.
    \label{thm:w1-multicolorclique}
\end{restatable}

As described in \cref{sec:mixed-hardness}, the variables of the ILP in \cref{thm:w1-multicolorclique} can be duplicated and linked through graph cycles to reduce the unary sized coefficients in $W$ and obtain a $0/1$ constraint matrix, yielding \cref{thm:wide-w1-hard}. As a result of this construction, the work of Marx~\cite{DBLP:journals/toc/Marx10}, reveals a lower bound conditional on ETH that rules out an $f(h)n^{o(h/\log h)}$ time algorithm to solve the instances from \cref{thm:wide-w1-hard}.

\section[Reduction to perfect b-matching]{Reduction to perfect $b$-matching}
\label{sec:master-reduction}

We proceed to present the details of \cref{sec:overview} and prove the required intermediate results. First, we consider \cref{lemma:master-reduction} which simplifies the design of the algorithms for \crefilp{ilp:mixed} by showing that it suffices to give algorithms exploiting backdoors to $b$-matching problems. To do so, we generalize the reduction steps in~\cite{schrijver2003combinatorial} that transform the generalized matching problem to the perfect $b$-matching problem. Most steps of \cref{lemma:master-reduction} can be interpreted as gadget constructions in bidirected graphs, see~\cite{schrijver2003combinatorial}, but we present them in terms of ILPs to retain a direct connection with the additional variables and linear constraints.

\lemmamasterreduction*

\begin{proof}
    We split the reduction into multiple steps where we restrict the generalized incidence matrix $M$ and the variable bounds on the variables $x$ corresponding to the matching. To accomplish this, we add new columns and rows to the matrix $M$ and add corresponding zero columns and rows to $T$ and $W$ respectively, unless mentioned otherwise. In each step, transformations are performed on the previous instances and the growth of the instance parameters are expressed in terms of the result of the previous step.
    \begin{claim*}
        All objective coefficients and columns of $A$ corresponding to the variables $x$ can be made nonnegative while additionally ensuring that $M'\in\{0,1\}^{m'\times n'}$ at the cost of increasing $n',m'=\O(n+m)$.
        \label{claim:nonnegativity}
    \end{claim*}
    \begin{claimproof}
        We split this into three steps.
        \proofsubparagraph{Inversion of nonpositive columns of $M$.}
        First, multiply any column $j$ for which $M_{\cdot,j}$ has a negative coefficient and no positive coefficients by $-1$ as shown in \cref{fig:invert-column}.
        \begin{figure}[H]
            \begin{cdisplaymath}
                \begin{array}{c}
                    c_j\\
                    \hdashline{}
                    A_{\cdot,j}\\
                    \hdashline{}
                    [l_j,u_j]
                \end{array}
                \to
                \begin{array}{c}
                    -c_j\\
                    \hdashline{}
                    -A_{\cdot,j}\\
                    \hdashline{}
                    [-u_j,-l_j]
                \end{array}
            \end{cdisplaymath}
            \caption{A column corresponding to an edge with total negative sign can be inverted to reduce the number of negative coefficients in the matching matrix.}
            \label{fig:invert-column}
        \end{figure}
        That is, we replace the $j$-th column $A_{\cdot,j}$ of $A$ with $-A_{\cdot,j}$, the domain $[l_j,u_j]$ of $x_j$ with $[-u_j,-l_j]$ and its objective coefficient $c_j$ with $-c_j$.
        
        \proofsubparagraph{Ensuring the nonnegativity of $c$ and $W$.}
        Second, we ensure that $c\ge\veczero,W\in\Z_{\ge0}^{h\times n}$ and, as a side product, additionally eliminate entries with a coefficient of $2$. To do so, we split the columns into their positive and negative parts as visualized in \cref{fig:sign-split}.
        \begin{figure}[H]
            \centering
            \begin{subfigure}{0.55\textwidth}
                \begin{cdisplaymath}
                    \begin{array}{c:l}
                        c_j\\
                        \cdashline{1-1}{}
                        W_{\cdot,j}&\\
                        \cdashline{1-1}{}
                        \veczero&\\
                        M_{i_1,j}&=b_{i_1}\\
                        \veczero&\\
                        M_{i_2,j}&=b_{i_2}\\
                        \veczero&\\
                        \\
                        \cdashline{1-1}{}
                        [l_j,u_j]&
                    \end{array}
                    \to
                    \begin{array}{c c:l}
                        (c_j)_+&-(c_j)_-&\\
                        \cdashline{1-2}{}
                        (W_{\cdot,j})_+&-(W_{\cdot,j})_-\\
                        \cdashline{1-2}{}
                        \veczero&\veczero\\
                        M_{i_1,j}&0&=b_{i_1}\\
                        \veczero&\veczero\\
                        0&-M_{i_2,j}&=b_{i_2}\\
                        \veczero&\veczero\\
                        1&1&=0\\
                        \cdashline{1-2}{}
                        [l_j,u_j]&[-u_j,-l_j]&
                    \end{array}
                \end{cdisplaymath}
                \caption{Splitting a column without a $2$ coefficient}
                \label{fig:sign-split-normal}
            \end{subfigure}\\
            \begin{subfigure}{0.55\textwidth}
                \begin{cdisplaymath}
                    \begin{array}{c:l}
                        c_j\\
                        \cdashline{1-1}{}
                        W_{\cdot,j}&\\
                        \cdashline{1-1}{}
                        \veczero&\\
                        2&=b_i\\
                        \veczero&\\
                        \\
                        \cdashline{1-1}{}
                        [l_j,u_j]&
                    \end{array}
                    \to
                    \begin{array}{c c:l}
                        (c_j)_+&-(c_j)_-&\\
                        \cdashline{1-2}{}
                        (W_{\cdot,j})_+&-(W_{\cdot,j})_-\\
                        \cdashline{1-2}{}
                        \veczero&\veczero\\
                        1&-1&=b_i\\
                        \veczero&\veczero\\
                        1&1&=0\\
                        \cdashline{1-2}{}
                        [l_j,u_j]&[-u_j,-l_j]&
                    \end{array}
                \end{cdisplaymath}
                \caption{Splitting a column with a $2$ coefficient}
                \label{fig:sign-split-self-loop}
            \end{subfigure}
            \caption{A column may be split into its positive and negative part to ensure that all the objective coefficients and coefficients corresponding to the additional global constraints are nonnegative.}
            \label{fig:sign-split}
        \end{figure}
        For this, let $a_+=\max\{0,a\}$ and $a_-=\min\{0,a\}$ denote the positive and negative part of some scalar $a$ and employ the same notation for coordinate-wise vector operations. For every column $j\in[n]$ such that $\|M_{\cdot,j}\|_\infty\le1$ for which its two nonzero entries are contained in $\{M_{i_1,j},M_{i_2,j}\}$ and $M_{i_1,j}\ge0$, we implement the following modifications:
        \begin{itemize}
            \item Split $x_j$ into the variables $x_j^{(+)},x_j^{(-)}$ where $x_j^{(+)}\in[l_j,u_j],x_j^{(-)}\in[-u_j,-l_j]$.
            \item Introduce the constraint $x_j^{(+)}+x_j^{(-)}=0$ so that $x_j^{(-)}=-x_j^{(+)}$.
            \item Split the original objective term $c_jx_j$ into $(c_j)_+x_j^{(+)}-(c_j)_-x_j^{(-)}$.
            \item Similarly replace $W_{\cdot,j}x_j$ with $(W_{\cdot,j})_+x_j^{(+)}-(W_{\cdot,j})_-x_j^{(-)}$.
            \item Replace the term $M_{i_1,j}x_j$ in the $i_1$-th constraint with $M_{i_1,j}x_j^{(+)}$.
            \item Replace the term $M_{i_2,j}x_j$ in the $i_2$-th constraint with $-M_{i_2,j}x_j^{(-)}$.
        \end{itemize}
        For a column with a single coefficient of $2=M_{i,j}$, we may apply the same procedure by interpreting $i_1=i_2=i,M_{i_1,j}=M_{i_2,j}=1$ and replacing the term $2x_j$ with $1x_j^{(+)}-1x_j^{(-)}$.
        
        \proofsubparagraph{Eliminating mixed sign columns in $M$.}
        Third, for a remaining column $M_{\cdot,j}$ with nonzero coefficients $-1$ and $1$ we again split $x_j$ into $x_j^{(+)}\in[l_j,u_j]$ and $x_j^{(-)}\in[-u_j,-l_j]$ as shown in \cref{fig:sign-split-mixed}.
        \begin{figure}[H]
            \begin{cdisplaymath}
                \begin{array}{c}
                    c_j\\
                    \hdashline{}
                    W_{\cdot,j}\\
                    \hdashline{}
                    \veczero\\
                    -1\\
                    \veczero\\
                    1\\
                    \veczero\\
                    \\
                    \hdashline{}
                    [l_j,u_j]
                \end{array}
                \to
                \begin{array}{c c:l}
                    c_j&0&\\
                    \cdashline{1-2}{}
                    W_{\cdot,j}&\veczero&\\
                    \cdashline{1-2}{}
                    \veczero&\veczero&\\
                    0&1&\\
                    \veczero&\veczero&\\
                    1&0&\\
                    \veczero&\veczero&\\
                    1&1&=0\\
                    \cdashline{1-2}{}
                    [l_j,u_j]&[-u_j,-l_j]
                \end{array}
            \end{cdisplaymath}
            \caption{An edge with mixed signs can be split into two positively signed edges which take opposing values.}
            \label{fig:sign-split-mixed}
        \end{figure}
        This encompasses the following:
        \begin{itemize}
            \item Replace $1x_j$ with $1x_j^{(+)}$ and $-1x_j$ with $1x_j^{(-)}$ in $M$.
            \item Add the constraint $x_j^{(+)}+x_j^{(-)}=0$.
            \item Copy the coefficients of $W$ and $c$ corresponding to $x_j$ over to $x_j^{(+)}$.
            \item Set the corresponding coefficients of $x_j^{(-)}$ to zero.
        \end{itemize}
        
        After this, all columns of $M$ have either no nonzero coefficients, a single $1$ coefficient or two $1$ coefficients. The size of the ILP in terms of $n+m$ is increased by a constant factor.
    \end{claimproof}
    \begin{claim*}
        All lower bounds $e,l$ can be set to $e',l'=\veczero$ resulting in $\|g'\|_1=\|g-e\|_1,\|u'\|_1=\|u-l\|_1,\|d'\|_1=\O(\|d\|_1+\|A\|_1(\|e\|_1+\|l\|_1)),\|b'\|_1=\O(\|b\|_1+\|A\|_1(\|e\|_1+\|l\|_1))$ and, if $p=0$, one may restrict $\|b'\|=\O(\|u-l\|_1)$.
    \end{claim*}
    \begin{claimproof}
        We translate the system by substituting $x'=x-l$. This yields new bounds $(e',l')=\veczero,(g',u')=(g,u)-(e,l)$ and right-hand side vector $(d',b')=(d,b)-A(e,l)$. The objective becomes a translation of the original. The bounds on $\|d'\|_1$ and $\|b'\|_1$ follow. For the special case where $p=0$, we may observe that $\|b'\|_1\le\|M\|_1\|u'\|_1\le2\|u-l\|_1=\O(\|u-l\|_1)$ holds or the system must be trivially infeasible.
    \end{claimproof}
    \begin{claim*}
        All columns of $M$ can be extended to full sum $\|M_{\cdot,j}\|_1=2$ by increasing $\|u'\|_1=\O(\|g\|_1+\|u\|_1)$ and $\|b'\|_1=\O(\|b_1\|+\|g\|_1+\|u\|_1)$, and adding a binary vector to $T$.
    \end{claim*}
    \begin{claimproof}
        First, we treat the columns with $\|M_{\cdot,j}\|_1=1$. For this purpose, add a new variable $s^{(1)}$ with zero objective coefficient that will occur in a new redundant constraint. Let $B=\begin{pmatrix}T&M\end{pmatrix}$ be the lower part of the constraint matrix $A$. Now, for every column $B_{\cdot,j}$ of which the sum of the entries is odd, we introduce a coefficient of one in the new constraint. More formally, we define the row vector $r^{(1)}$ by $r^{(1)}=\vecone^\top B\bmod2$ and add the constraint $2s^{(1)}+r^{(1)}(y,x)=t^{(1)}$ where $t^{(1)}=(\vecone^\top b\bmod2)+2\lceil r^{(1)}(g,u)/2\rceil$ to the ILP. Here, $\bmod\,2$ denotes the component-wise remainder of the vector modulo $2$. By assigning $s^{(1)}$ a sufficiently large domain and ensuring that all integral solutions $x$ to the original problem instance satisfy that $t^{(1)}-r^{(1)}(y,x)\equiv\vecone^\top b-r^{(1)}(y,x)\equiv0\pmod2$, this constraint becomes redundant and preserves the equivalence of the problems. For this, observe that $\vecone^\top B(y,x)=\vecone^\top b$ when $(y,x)$ is a feasible solution, which shows that $\vecone^\top b-r^{(1)}(y,x)=\vecone^\top B(y,x)-r^{(1)}(y,x)\equiv0\pmod2$. By construction, it is admissible to restrict $s^{(1)}\in[0,\lceil r^{(1)\top}(g,u)/2\rceil]$, which increases $\|u\|_1$ by at most $\|g\|_1+\|u\|_1$.

        A similar technique can be used to treat the zero columns of $M$. Define the row vector $r^{(2)}$ by $r_j^{(2)}=2$ if $j$ corresponds to a column of $M$ and $M_{\cdot,j}=\veczero$, and $r_j^{(2)}=0$ otherwise. We may add the redundant constraint $2s^{(2)}+r^{(2)}(y,x)=2\|u\|_1$ and restrict $s^{(2)}\in[0,\|u_1\|]$. The coefficients of $2$ in $M$ introduced in this step can be eliminated without increasing the asymptotic size of the ILP as shown in the earlier step which splits the columns.
    \end{claimproof}
    \begin{claim*}
        One can set $u'=\vecinfty$ and restrict $M$ to be the incidence matrix of a simple graph while only increasing $\|b'\|_1=\O(\|b\|_1+\|u\|_1)$.
    \end{claim*}
    \begin{claimproof}
        The edge capacities are removed by essentially subdividing every edge into three edges and two intermediate vertices as shown in \cref{fig:capacity-elimination}.
        \begin{figure}[H]
            \begin{cdisplaymath}
                \begin{array}{c:l}
                    c_j&\\
                    \cdashline{1-1}{}
                    W_{\cdot,j}&\\
                    \cdashline{1-1}{}
                    \veczero&\\
                    1&=b_{i_1}\\
                    \veczero&\\
                    1&=b_{i_2}\\
                    \veczero&\\
                    &\\
                    &\\
                    \cdashline{1-1}{}
                    [0,u_j]&
                \end{array}
                \to
                \begin{array}{c c c:l}
                    c_j&0&0&\\
                    \cdashline{1-3}{}
                    W_{\cdot,j}&\veczero&\veczero\\
                    \cdashline{1-3}{}
                    \veczero&\veczero&\veczero\\
                    1&0&0&=b_{i_1}\\
                    \veczero&\veczero&\veczero\\
                    0&0&1&=b_{i_2}\\
                    \veczero&\veczero&\veczero\\
                    1&1&0&=u_j\\
                    0&1&1&=u_j\\
                    \cdashline{1-3}{}
                    [0,\infty)&[0,\infty)&[0,\infty)&
                \end{array}
            \end{cdisplaymath}
            \caption{Variable upper bounds can be eliminated by modeling the slack of the constraint with an intermediate variable.}
            \label{fig:capacity-elimination}
        \end{figure}
        To accomplish this, we implement the following modifications:
        \begin{itemize}
            \item Split every variable $x_j$ with a coefficient $1$ in rows $i_1$ and $i_2$ into $x_j^{(1)},x_j^{(s)},x_j^{(2)}$, each with lower bound $0$ and no upper bound.
            \item Variable $x_j^{(1)}$ copies the column $W_{\cdot,j}$ and objective coefficient $c_j$ from $x_j$.
            \item The other variables $x_j^{(s)},x_j^{(2)}$ have zeroes in their respective parts of $W$ and $c$.
            \item Add the constraints $x_j^{(1)}+x_j^{(s)}=u_j$ and $x_j^{(2)}+x_j^{(s)}=u_j$ so that $x_j^{(s)}$ models the slack of the capacity constraint $x_j\le u_j$ and $x_j^{(1)}=x_j^{(2)}$.
            \item Replace the term $x_j$ with $x_j^{(1)}$ in row $i_1$ and with $x_j^{(2)}$ in row $i_2$.
        \end{itemize}
    \end{claimproof}

    The bounds on the coefficients can be derived by following the reduction steps taken and suitably substituting the size increments.
\end{proof}

Having shown that the instances of \crefilp{ilp:general} may be restricted to be a perfect $b$-matching problem with additional complicating variables or constraints, we may rely on this in the algorithm constructions in \cref{sec:tall,sec:mixed-algorithm}.

\section{Few arbitrary variables}
\label{sec:tall}

This section is devoted to proving \cref{thm:tall-fpt} as outlined in \cref{sec:overview-few-arbitrary-variables}, which shows that ILP is FPT when parameterized by the number of columns with $1$-norm larger than $2$.

\thmtallfpt*

Recall that our strategy to solve \crefilp{ilp:tall} is to guess the remainder of the $p$ complicating variables modulo $2$ and then solve an ILP on only these $p$ variables by modeling the remaining matching part of the problem using a polyhedral constraint. For this, we require the convexity result from \cref{lemma:convexity-of-b-matching}. Recall \cref{def:f,def:sbo-jump-m-convex}. First, we modify the alternating path argument of Kobayashi~\cite{DBLP:conf/ipco/Kobayashi23} to show that the function $f_{c,M}$ which models the cost of a minimum cost $b$-matching is SBO jump M-convex. Then, we show that an arbitrary SBO jump M-convex function $f$ satisfies the required lattice convexity on $2\Z^m+r$. The essence of Kobayashi's argument is to consider two optimal matchings and decompose their difference into alternating paths. The endpoints of these paths then yield a suitable $2$-step decomposition $p^{(k)}$ and the changes they induce in the cost of the matchings yield the required objective steps $g^{(k)}$. As Kobayashi's argument works with $\vecone$-capacitated $b$-matchings, we need a slightly different approach to obtain suitable alternating paths. For this purpose, we present a construction used to pseudo-polynomially reduce the perfect $b$-matching problem to a perfect matching problem in \cref{prop:gb}~\cite{schrijver2003combinatorial}. The construction expands the vertices and edges of the original graph into multiple copies so that assigning a value of $x_e$ to $e$ in a $b$-matching in $G$ corresponds to choosing $x_e$ copies of $e$ in the constructed graph. This is visualized in \cref{fig:G_b}.

\begin{proposition}[Section 31.1 in~\cite{schrijver2003combinatorial}]
    Let $b,\overline b\in\mathbb Z_{\ge0}^V$ satisfy $b\le\overline b$ and $G=(V,E)$ be a simple graph. Let $G_{\overline b}$ be the graph defined by creating $\overline b_v$ copies of each vertex $v\in V$ and creating an edge between every pair of copies of $v_1$ and $v_2$ when $v_1$ and $v_2$ are adjacent in $G$. Then $G$ has a perfect $b$-matching if and only if $G_{\overline b}$ has a matching saturating exactly $b_v$ copies of $v$ for each $v\in V$. In particular, a given $x$ is a perfect $b$-matching of $G$ if and only if $G_{\overline b}$ has a matching $M$ such that $M$ contains exactly $x_{\{v_1,v_2\}}$ edges between the copies of $v_1$ and $v_2$ and saturates exactly $b_v$ copies of $v$.
    \label{prop:gb}
\end{proposition}

\begin{figure}
    \begin{subfigure}{0.45\textwidth}
        \centering
        \includegraphics[width=0.8\textwidth]{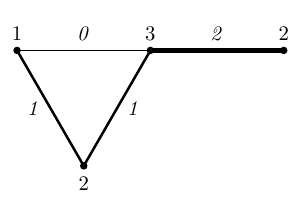}
        \caption{$G$ with $b$ specified on the vertices and a perfect $b$-matching specified by the edge labels}
    \end{subfigure}
    \hspace{0.05\textwidth}
    \begin{subfigure}{0.45\textwidth}
        \centering
        \includegraphics[width=0.8\textwidth]{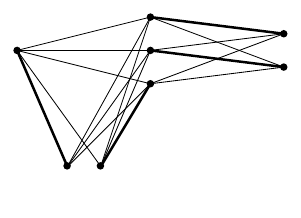}
        \caption{$G_b$ with a perfect matching $M$ corresponding to the perfect $b$-matching in $G$}
    \end{subfigure}
    \caption{A pseudo-polynomial reduction construction between matching and $b$-matching~\cite{schrijver2003combinatorial}.}
    \label{fig:G_b}
\end{figure}

By using this construction for an appropriate $\overline b$, we can use the alternating paths appearing in Kobayashi's~\cite{DBLP:conf/ipco/Kobayashi23} argument in $G_{\overline b}$ to provide a suitable $2$-step decomposition for points in the domain of $f_{c,M}$.

\begin{proposition}
    The function $f_{c,M}(z)=\min\bigl\{c^\top x\bigm\vert Mx=z,x\in\Z_{\ge0}^n\bigr\}$ is SBO jump M-convex.
    \label{prop:b-matching-is-sbo-jump-m-convex}
\end{proposition}

\begin{proof}
    Let $z^{(1)},z^{(2)}\in\Z^m$ be arbitrary integer points in the domain of $f_{c,M}$. For $k=1,2$, let $x^{(k)}$ be a perfect $z^{(k)}$-matching in $G:=G(M)$ with minimum cost $f_{c,M}(z^{(k)})$. Consider arbitrary associated matchings $M^{(1)}$ and $M^{(2)}$ in $G_{\overline z}$ where $\overline z=\max\{z^{(1)},z^{(2)}\}$ is the coordinate-wise maximum of $z^{(1)}$ and $z^{(2)}$. The spanning subgraph with the symmetric difference between $M^{(1)}$ and $M^{(2)}$ as edge set can be uniquely decomposed into disjoint paths and cycles of which the edges alternate between $M^{(1)}$ and $M^{(2)}$. We associate an alternating walk $W$ with a vector $d\in\Z^n$, defined over the edges $[n]$ in the original graph, so that it corresponds to the change in the $b$-matching $x$ when adding edges to and removing edges from the matching in $G_{\overline z}$ along this walk. For this purpose, we define $d$ by setting $d_i$ to be the number of edges that are copies of $i$ in $W\cap M^{(2)}$ minus the number of copies in $W\cap M^{(1)}$.
    
    In this way, the difference $x^{(2)}-x^{(1)}$ decomposes as $x^{(2)}-x^{(1)}=\sum_kd^{(k)}$, where $k$ ranges over all decomposed paths and cycles. Similarly, when defining $p^{(k)}=Md^{(k)}$ we find a decomposition $z^{(2)}-z^{(1)}=\sum_kMd^{(k)}$. As an alternating cycle induces a difference vector $d$ with $Md=0$, they may be left out of the decomposition of $z^{(2)}-z^{(1)}$. For this reason, assume that the indices $1,\dots,\ell$ index the difference vectors associated with the alternating paths and consider the decomposition $z^{(2)}-z^{(1)}=\sum_{k\in[\ell]}p^{(k)}$. From the fact that a vector $d^{(k)}$ is induced by an alternating path, it follows that $p^{(k)}$ has a $1$-norm of zero or two. We may observe that $d^{(k)}$ cannot be the zero vector as this would imply that $z_i^{(1)}<\overline z_i$ and $z_i^{(2)}<\overline z_i$ for the vertex $i$ in $G$ that corresponds to the beginning and end vertices of the path. This is impossible due to the tight choice of $\overline z$. An analogous argument shows that each step $d^{(k)}$ is sign-compatible with $z^{(2)}-z^{(1)}$, which allows us to conclude that $p^{(1)},\dots,p^{(\ell)}$ is indeed a two-step decomposition of $z^{(2)}-z^{(1)}$.

    This also suggests to define $g^{(k)}$ as $c^\top d^{(k)}$. Note that $c^\top d$ must be $0$ for a difference vector $d$ induced by an alternating cycle, because if it were not, one of $M^{(1)}$ or $M^{(2)}$ is not of minimum cost. This shows that $f_{c,M}(z^{(2)})-f_{c,M}(z^{(1)})=\sum_{k\in[\ell]}g^{(k)}$. Note that, for any $I\subseteq[\ell]$, the vector $x^{(1)}+\sum_{k\in I}d^{(k)}$ is a proper perfect $(z^{(1)}+\sum_{k\in I}p^{(k)})$-matching with cost $f_{c,M}(z^{(1)})+\sum_{k\in I}g^{(k)}$, showing that $f_{c,M}(z^{(1)}+\sum_{k\in I}p^{(k)})$ is bounded from above as required.
\end{proof}

We now proceed to show the lattice-convexity of SBO jump M-convex functions. The proof is split into two steps. \cref{lemma:sbo-jump-m-convex-closing} shows how two points in a convex combination can be pairwise modified without increasing the overall sum of function values. This operation is then exhaustively employed in the proof of the final lattice-convexity result.

\begin{lemma}
    Let $f$ be an SBO jump M-convex function. Let $z^{(1)},z^{(2)}\equiv r\pmod2$ and let $i^*\in[m]$ be so that $z_{i^*}^{(1)}-z_{i^*}^{(2)}\ge2$. Then, there exist $z^{(1)\prime}$ and $z^{(2)\prime}$ such that
    \begin{itemize}
        \item $z^{(1)\prime},z^{(2)\prime}\equiv r\pmod2$,
        \item $z^{(1)\prime}+z^{(2)\prime}=z^{(1)}+z^{(2)}$,
        \item $z_{i^*}^{(1)\prime}=z_{i^*}^{(1)}-2,z_{i^*}^{(2)\prime}=z_{i^*}^{(2)}+2$,
        \item $z_i^{(1)}=z_i^{(2)}\implies z_i^{(1)}=z_i^{(1)\prime}=z_i^{(2)\prime}=z_i^{(2)}$ for all $i\in[m]$,
        \item $f(z^{(1)\prime})+f(z^{(2)\prime})\le f(z^{(1)})+f(z^{(2)})$.
    \end{itemize}
    \label{lemma:sbo-jump-m-convex-closing}
\end{lemma}

\begin{proof}
    W.l.o.g., consider the case where $z^{(1)}\ge z^{(2)}$. The general case is equivalent up to changing signs. Consider a $2$-step decomposition $p^{(1)},\dots,p^{(\ell)}$ of $z^{(2)}-z^{(1)}$ and $g\in\R^\ell$ from \cref{def:sbo-jump-m-convex}. We create a multigraph representation $G$ of this decomposition with vertex set $[m]$. For every step $p^{(k)}$ that has a $-1$ component at coordinates $i_1$ and $i_2$, add an edge that connects $i_1$ and $i_2$. For a step $p^{(k)}$ for which $p_i^{(k)}=-2$, add a self-loop on $i$. Since each vertex in $G$ has an even degree and $i^*$ is not an isolated vertex by assumption, we can pick an arbitrary cycle that contains $i^*$. Let $I\subseteq[\ell]$ be the indices of the steps corresponding to the edges of the cycle, which defines an even step from $z^{(1)}$ to $z^{(2)}$. We define $z^{(1)\prime}=z^{(1)}+\sum_{k\in I}p^{(k)}$ and $z^{(2)\prime}=z^{(2)}-\sum_{k\in I}p^{(k)}$. Note that the step $\sum_{k\in I}p^{(k)}$ is in $\{-2,0\}^m$ because the degree of each vertex in the cycle is $0$ or $2$. In particular, the degree of $i^*$ is exactly $2$. We conclude that the first, second and third properties are satisfied. The satisfaction of the fourth property follows from that $p_i^{(k)}=0$ when $z_i^{(1)}=z_i^{(2)}$. Finally, observe that $f(z^{(1)\prime})+f(z^{(2)\prime})\le f(z^{(1)})+\sum_{k\in I}g^{(k)}+f(z^{(1)})+\sum_{k\in[\ell]\setminus I}g^{(k)}=f(z^{(1)})+f(z^{(2)})$, showing that the fifth requirement is satisfied.
\end{proof}

We can now exhaustively employ \cref{lemma:sbo-jump-m-convex-closing} pairwise on vectors of a given convex combination to derive the lattice-convexity of $f$. \cref{lemma:convexity-of-sbo-jump-m-convex-functions} generalizes \cref{lemma:convexity-of-b-matching}.

\begin{lemma}
    Let $f$ be an SBO jump M-convex function. Let $z^{(1)},z^{(2)},\dots,z^{(\ell)}\equiv r\pmod2$ be given and $\lambda^{(1)},\lambda^{(2)},\dots,\lambda^{(\ell)}\ge0$ be real convex multipliers, i.e., $\sum_{k\in[\ell]}\lambda^{(k)}=1$, such that $z:=\sum_{k\in[\ell]}\lambda^{(k)}z^{(k)}\equiv r\pmod2$. Then $f(z)\le\sum_{k\in[\ell]}\lambda^{(k)}f(z^{(k)})$.
    \label{lemma:convexity-of-sbo-jump-m-convex-functions}
\end{lemma}

\begin{proof}
    Since we may assume that $\lambda$ is rational, we may additionally assume that $\lambda_k=1/\ell$ after duplicating vectors of the convex combination a suitable number of times.

    For every $i^*\in[m]$ we exhaustively employ the following procedure: if there exist $z^{(k_1)},z^{(k_2)}$ with $z_{i^*}-z_{i^*}^{(k_1)}\le-2$ and $z_{i^*}-z_{i^*}^{(k_2)}\ge2$, apply \cref{lemma:sbo-jump-m-convex-closing} to replace $z^{(k_1)}$ and $z^{(k_2)}$ with the resulting $z^{(k_1)\prime}$ and $z^{(k_2)\prime}$ from this lemma. This preserves the total sum $\sum_{k\in[\ell]}z^{(k)}=\ell z$ and guarantees that $\sum_{k\in[\ell]}f(z^{(k)})$ does not increase. Observe that if such $z^{(k_1)}$ and $z^{(k_2)}$ do not exist, w.l.o.g. all $z_{i^*}^{(k)}\ge z_{i^*}$. As $z_{i^*}$ is the average of $z_{i^*}^{(k)}$, it must hold that $z_{i^*}=z_{i^*}^{(k)}$. Since the application of \cref{lemma:sbo-jump-m-convex-closing} does not modify any $z_i^{(k)}$ at $i$ for which $z_i^{(k)}=z_i$ for all $k\in[\ell]$, exhaustive application will eventually result in $z^{(k)\prime}=z$.
    
    It is left to show that \cref{lemma:sbo-jump-m-convex-closing} can be exhaustively applied. To see this, we employ $\sum_{k\in[\ell]}|z_{i^*}-z_{i^*}^{(k)}|$ as progress measure. With respect to the progress made after applying \cref{lemma:sbo-jump-m-convex-closing}, it holds that
    \begin{align*}
        &\sum_{k\in[\ell]}\bigl|z_{i^*}-z_{i^*}^{(k)\prime}\bigr|-\sum_{k\in[\ell]}\bigl|z_{i^*}-z_{i^*}^{(k)}\bigr|\\
        &=
        \bigl|z_{i^*}-z_{i^*}^{(k_1)\prime}\bigr|+\bigl|z_{i^*}-z_{i^*}^{(k_2)\prime}\bigr|-\bigl|z_{i^*}-z_{i^*}^{(k_1)}\bigr|+\bigl|z_{i^*}-z_{i^*}^{(k_2)}\bigr|\\
        &<0,
    \end{align*}
    which follows from \cref{lemma:sbo-jump-m-convex-closing}  that guarantees that $z_{i^*}^{(k_1)}$ and $z_{i^*}^{(k_2)}$ both move strictly closer to $z_{i^*}$. After applying \cref{lemma:sbo-jump-m-convex-closing} exhaustively for all $i^*\in[m]$, we end up with a situation where $z^{(k)\prime}=z$ for all $k$ and that $\ell f(z)=\sum_{k\in[\ell]}f(z^{(k)\prime})\le\sum_{k\in[\ell]}f(z^{(k)})$.
\end{proof}

We note that the lattice-convexity of the domain $\{z\in\Z^m\colon f_{c,M}(z)<\infty\}$ on $2\Z^m+r$ immediately follows from a generalization of Tutte's characterization of graphs with perfect matchings, see Corollary 31.1a in~\cite{schrijver2003combinatorial}. Alternatively, it is a special case of Theorem 1.2 in~\cite{DBLP:journals/jgt/AnsteeN99}.

Having established the convexity of $f_{c,M}$ on the scaled and shifted lattice $2\Z^m+r$, we find that, for a fixed $r$, the set of $z$-s with corresponding perfect $z$-matchings and its optimal objective can be modeled as the points in a convex set in $1+m$ dimensions. 

We define a polyhedron in \cref{def:pr}, which models this set within some prescribed bound $U$ and can be interpreted as the epigraph of a convex extension of $f_{c,M}$.

\begin{definition}
    Let $r\in\{0,1\}^m$ be a remainder vector modulo $2$ and $U\in\Z$ an upper bound with polynomial encoding length. We define the polyhedron $P_{r,U}$ to be the convex hull of the points $S_{r,U}$ defined by
    \[
        S_{r,U}:=\bigl\{(\omega,z)\in\R\times\Z^m\colon z\equiv r\pmod2,\left\|z\right\|_\infty\le U,f_{c,M}(z)\le\omega\bigr\}.
    \]
    That is, $P_{r,U}$ is the Minkowski sum of the polytope which is the convex hull of the points
    \[
        \tilde S_{r,U}=\bigl\{(f_{r,U}(z),z)\bigm\vert z\equiv r\pmod2,\left\|z\right\|_\infty\le U,f_{c,M}(z)<\infty,z\in\Z^m\bigr\}
    \]
    and the recession cone $\{(\lambda,\veczero)\ \vert\ \lambda\ge0\}$.
    \label{def:pr}
\end{definition}

That is, $S_{r,U}$ contains points $(\omega,z)$ for which there is a perfect $z$-matching with cost at most $\omega$, restricted to remainder vector $r$ modulo $2$ and a bounding box. Note that we do not have an explicit compact outer description of $P_{r,U}$ and are not aware of whether a polyhedron that models $f_{c,M}$ on $2\Z^m+r$ with polynomially many facets that are efficiently computable exists. However, such a description is not needed, as known integer programming algorithms, e.g.~\cite{DBLP:conf/focs/ReisR23,DBLP:journals/mor/Lenstra83}, can solve integer programs when the feasible region is a bounded polyhedron equipped with a strong separation oracle. Similar as to what was observed in~\cite{DBLP:journals/mor/Zhang03} for a slightly different polytope, the known methods for optimizing generalized matchings, imply that such a strong separation oracle running in polynomial time exists as a consequence of the ellipsoid method.

\begin{lemma}
    The strong separation problem for $P_{r,U}$ can be solved in polynomial time. That is, given a fractional $\tilde q\in\Q^{1+m}$ one can efficiently decide whether $\tilde q\in P_{r,U}$ and if not, yield a separating hyperplane $\alpha\in\Q^{1+m}$ such that $\alpha^\top\tilde q>\alpha^\top q$ for all $q\in P_{r,U}$.
    \label{lemma:separation-oracle}
\end{lemma}

\begin{proof}
    By construction, $P_{r,U}$ is a polyhedron with its vertices having $z$ bounded by $U$ and consequently additionally having suitably bounded $\omega$. As its recession cone is similarly well-behaved, we can apply a deep result following from the ellipsoid method, Theorem 6.4.9 in~\cite{DBLP:books/sp/GLS1988}, which shows that the strong separation problem can efficiently be solved if one can optimize an arbitrary objective $\tilde c$ over $P_{r,U}$ in polynomial time.

    For completeness, we explicitly show how one can optimize over $P_{r,U}$. It suffices to solve the following problem
    \[
        \min\left\{\tilde c^\top(\omega,z)\bigm\vert(\omega,z)\in P_{r,U}\right\}.
    \]
    If the first component of $\tilde c$ is negative, the problem is unbounded in the $(1,\veczero)$ direction if and only if it is feasible. Therefore, this case can be reduced to the same problem with $\tilde c=\veczero$ and we may further assume that the first component of $\tilde c$ is nonnegative. In this case, it suffices to optimize $\tilde c$ over the set $\tilde S_{r,U}$, which contains all vertices of $P_{r,U}$. Expanding the definitions of $\tilde S_{r,U}$ and $f(z)$ yields
    \[
        \min\bigl\{\tilde c^\top(c^\top x,z)\bigm\vert z\equiv r\pmod2,\|z\|_\infty\le U,Mx=z,x\in\Z_{\ge0}^n,z\in\Z^m\bigr\}.
    \]
    Substituting $z=Mx$ yields
    \[
        \min\left\{\tilde c^\top(c^\top x,Mx)\bigm\vert Mx\equiv r\pmod2,Mx\le U\vecone,x\in\Z_{\ge0}^n\right\}.
    \]
    This is an instance of a polynomial time solvable matching generalization (\cite{DBLP:journals/mp/EdmondsJ73}, Theorem 36.5 in~\cite{schrijver2003combinatorial}). For completeness, we further reduce it with an additional step. Observe that after defining $\tilde r:=(U\vecone+r)\bmod2$, we can reformulate the constraints on $Mx$ by equating $Mx$ to $U\vecone-(\tilde r+2Is)$ for some auxiliary integer variables $s\in\Z_{\ge0}^m$. This yields the generalized matching problem
    \[
        \min\left\{\bigl((\tilde c^\top(c^\top,M))^\top\bigr)^\top x\bigm\vert Mx+2Is=U\vecone-\tilde r,(x,s)\in\Z_{\ge0}^{n+m}\right\}.
    \]
    By \cref{thm:generalized-matching-in-p}, this problem can be solved in polynomial time.
\end{proof}

We have now gathered the necessary ingredients to prove \cref{thm:tall-fpt}.

\begin{proof}[Proof of \cref{thm:tall-fpt}]
    Consider an instance of \crefilp{ilp:tall}. We first perform some preprocessing steps. Using the proximity result from Cook et al.~\cite{DBLP:journals/mp/CookGST86}, we may assume that the variable bounds $e,l,g,u$ are finite and have polynomial encoding length. Now apply \cref{lemma:master-reduction} to transform the ILP. The resulting problem is of the form
    \[
        \min\Bigl\{a^\top y+c^\top x\bigm\vert Ty+Mx=b,\veczero\le y\le g,y\in\Z^h,x\in\mathbb Z_{\ge0}^n\Bigr\},
    \]
    where $G(M)$ is simple and the bounds $g$ are finite with polynomial encoding length. Since this is a polynomial transformation of the bounded problem, which has an optimal objective with polynomial encoding length, an optimal solution to this reduced ILP can be found in a polynomial number of binary search steps solving a feasibility problem of the form
    \begin{equation}
        \Bigl\{a^\top y+c^\top x\le\omega^*,Ty+Mx=b,\veczero\le y\le g,y\in\Z^p,x\in\mathbb Z_{\ge0}^n\Bigr\}.
        \label{ilp:tall-clean}
    \end{equation}

    It now suffices to show that the feasibility of \crefilp{ilp:tall-clean} can be determined in FPT time. To do so, we guess the value of $y\bmod2$ by substituting $y=2v+t$ for a set of variables $v\in\Z^n$ and guessing all $2^p$ remainder vectors $t\in\{0,1\}^p$ as done in~\cite{DBLP:conf/soda/CslovjecsekKLPP24}. The \crefilp{ilp:tall-clean} is feasible if and only if at least one of these subproblems is feasible. After fixing $t$, our subproblem is equivalent to
    \begin{equation}
        \Bigl\{a^\top(2v+t)+c^\top x\le\omega^*,T(2v+t)+Mx=b,\veczero\le2v+t\le g,v\in\Z^p,x\in\mathbb Z_{\ge0}^n\Bigr\}
        \label{ilp:tall-d2}
    \end{equation}
    We show that this may equivalently be posed as finding an integer vector $v$ in a particular convex body. For this purpose, define the affine function $F$ by
    \[
        F(v):=\begin{pmatrix}
            \hat\omega(v)\\
            \hat z(v)
        \end{pmatrix}:=\begin{pmatrix}
            -2a^\top v-a^\top t+\omega^*\\
            -2Tv-Tt+b
        \end{pmatrix},
    \]
    define the requirement $R$ that there exists a suitably low costed perfect $\hat z(v)$-matching by
    \[
        R:=\left\{v\in\Z^p\colon\exists x\in\Z_{\ge0}^n:c^\top x\le\hat\omega(v),Mx=\hat z(v)\right\},
    \]
    and define the box $B$ expressing the variable bounds on $v$ by
    \[
        B:=\bigl\{v\in\R^p\colon\veczero\le 2v+t\le g\bigr\}.
    \]
    In this way, we may look for $v\in R\cap B$ to determine the feasibility of \crefilp{ilp:tall-d2}. In fact, if such $v$ exists, it can straightforwardly be completed to a solution to \crefilp{ilp:tall-clean} by setting $y=2v+t$ and computing a minimum cost perfect $\hat z(v)$-matching $x$.
    
    Since we fixed our remainder modulo $2$, the function $F$ satisfies that $\hat z(v)\equiv r\pmod2$ for all $v\in\Z^p$ where $r:=(-Tt+b)\mod2$, which will enable us to find such $v$ if it exists. Set $U$ to a value with polynomial encoding length that is sufficiently large so that $\|\hat z(v)\|_\infty\le U$ for all $v\in B$.

    \begin{claim*}
        $R\cap B=F^{-1}(P_{r,U})\cap B\cap\Z^p$, where $F^{-1}$ is the pre-image $\{v\in\R^p\colon F(v)\in P_{r,U}\}$.
    \end{claim*}
    \begin{claimproof}
        An integral $v\in F^{-1}(P_{r,U})\cap B$ is within the bounds $B$ and satisfies that $(\hat\omega(v),\hat z(v))$ is a convex combination $\sum_{k\in[\ell]}\lambda^{(\ell)}(\omega^{(k)},z^{(k)})$ of points in $S_{r,U}$. \cref{lemma:convexity-of-sbo-jump-m-convex-functions} shows that there exists a perfect $\hat z(v)$-matching $x$ with $c^\top x=f_{c,M}(\hat z(v))\le\sum_{k\in[\ell]}\lambda^{(k)}f_{c,M}(z^{(k)})\le\sum_{k\in[\ell]}\lambda^{(k)}\omega^{(k)}=\hat\omega(v)$. This witnesses that $v\in R$. The inclusion in the other direction makes use of the fact that $v\in R\cap B$ satisfies that $\|\hat z(v)\|_\infty\le U$. The perfect $\hat z(v)$-matching $x$ witnessing $v\in R$ shows that $f_{c,M}(\hat z(v))\le c^\top x\le\hat\omega(v)$ and thus that $F(v)=(\hat\omega(v),\hat z(v))\in S_{r,U}\subseteq P_{r,U}$.
    \end{claimproof}

    An integral point in the polytope $F^{-1}(P_{r,U})\cap B\cap\Z^p$ can be found in FPT time by an integer programming algorithm such as the algorithm by Reis and Rothvoss~\cite{DBLP:conf/focs/ReisR23}\footnote{This state-of-the-art IP algorithm based on Dadush work~\cite{dadush2012integer} is randomized. However, the algorithm by Lenstra~\cite{DBLP:journals/mor/Lenstra83} shows that the integer programming problem can be solved deterministically, yielding a slower, but deterministic FPT algorithm.}. For this, it is essential that one can implement a strong separation oracle over this polytope that runs in polynomial time, which follows from \cref{lemma:separation-oracle}. To see this, first observe that testing whether a given fractional $\tilde v$ is in $F^{-1}(P_{r,U})\cap B$ is equivalent to testing whether $F(\tilde v)\in P_{r,U}$ and $\tilde v\in B$. Hyperplanes separating $\tilde v$ from $B$ are trivial to construct. Finally, a hyperplane $\alpha\in\Q^{1+m}$ separating $F(\tilde v)$ from $P_{r,U}$ yields a hyperplane separating $\tilde v$ from $F^{-1}(P_{r,U})$ as $(\alpha^\top L)\tilde v>(\alpha^\top L)v$ for $v\in F^{-1}(P_{r,U})$ where $L\in\Z^{(1+m)\times p}$ is the linear transformation that comprises the affine $F$.

    If we find a feasible ILP, the integer programming algorithm can yield us a corresponding satisfying $v$.
\end{proof}

We note that the additional multiplicative $2^p$ contribution to the running time of the algorithm as a result of guessing the parity of $y$ is asymptotically small compared to the running time of the current state-of-the-art IP algorithm by Reis and Rothvoss~\cite{DBLP:conf/focs/ReisR23}, which for our application may have a running time of $(\log p)^{\O(p)}$ times a polynomial of the input encoding length.

\section{Few arbitrary variables and constraints}
\label{sec:mixed}

We now turn our focus on the more general \crefilp{ilp:mixed}, which may have additional complicating constraints as well as complicating variables. As mentioned, if the coefficients of the constraint matrix are encoded in binary, \crefilp{ilp:wide} can encode the NP-hard subset sum problem. Therefore, we develop an algorithm scaling polynomially in the unary encoding length of the maximum absolute value of the entries of the constraint matrix $\Delta$ and of the objective vector $\|c\|_\infty$ for a fixed $p$ and $h$.

\thmmixedxp*

Note that an instance of \crefilp{ilp:mixed} may be preprocessed so that $m+h\le n+p$ through Gaussian elimination within the running time of \cref{thm:mixed-xp}.

\subsection{Bounding the proximity of matching-like ILPs} \label{sec:proximity-bound}

Recall that the most important ingredient in deriving \cref{thm:mixed-xp} is a bound on the $\infty$-norm of the circuits of the constraint matrix. We prove \cref{lemma:circuit-ub-mixed} in this section.

\lemmacircuitubmixed*

We modify a well-known Cramer's rule based argument~\cite{onn2010nonlinear} to provide the bound of \cref{lemma:circuit-ub-mixed}. The proof of \cref{lemma:circuit-ub-mixed} starts by studying a circuit by a vector $\tilde c\in\Z^s$ obtained through Cramer's rule, of which each entry is a determinant of a submatrix of $A$. Here $s\le r$. Since subdeterminants of the incidence matrix of a nonbipartite graph can be exponential in $n$, we must estimate the greatest common divisor of the entries of $\tilde c$ in order to show that the corresponding normalized circuit is sufficiently small. The proof is set up as follows:
\begin{itemize}
    \item We can compute an entry of $\tilde c$ by Laplace expanding along the complicating columns and rows and obtain that $\tilde c$ is the signed sum of at most $\O(\Delta s)^{p+h}$ determinants of incidence submatrices $\tilde M$.
    \item All of these submatrices are similar in that they arise from a common submatrix $\hat M$ of $M$ after deleting at most $\O(p+h)$ columns and rows.
    \item The common matrix $\hat M$ can be rearranged to a block-diagonal form, which corresponds to decomposing the corresponding graph into connected components. Therefore, the determinant of $\hat M$ factorizes as the product of the determinants of the incidence matrices of the connected components.
    \item Every $\tilde M$ factorizes analogously. These factorizations contain all but a bounded number of components that appear in $\hat M$.
    \item Generalizing the rank-based argument used in \cite{DBLP:conf/aaai/BrandKO21}, we similarly observe that almost all blocks in the block diagonal arrangement of $\hat M$ have square dimensions.
    \item From the work of Grossman et al.~\cite{GROSSMAN1995213} it is known that the absolute value of the determinant of a nonsingular incidence matrix is a power of $2$.
    \item Therefore, if $\det\tilde M$ is large, it must be divisible by $2^z$ for some large exponent $z$, which reveals the presence of many square connected components with determinant $\pm2$ in $\hat M$.
    \item We can pick $z$ so that at least $z$ such components appear in all the submatrices $\tilde M$ obtained from the Laplace expansions, showing that the corresponding determinants are divisible by $2^z$.
    \item Using that also all nonsquare components are almost-square and the well known bound on determinants of incidence matrices in terms of the odd cycle packing number~\cite{GROSSMAN1995213} then shows that the subdeterminant of the remaining components in the factorizations of $\tilde M$, that are not one of the $z$ ``common'' components, is bounded.
    \item Therefore, by dividing all entries of $\tilde c$ by $2^z$, the integral $\tilde c\in\Z^s$ can be normalized to an integer vector bounded by $\O(\Delta s)^{p+h}$.
\end{itemize}

As suggested, the proof of \cref{lemma:circuit-ub-mixed} uses the insights from the work of Grossman et al.~\cite{GROSSMAN1995213}, which studies the subdeterminants of incidence matrices of simple graphs. Their results translate straightforwardly to bidirected graphs. To make this explicit, we first introduce the required notions.

For our purpose, we say that a bidirected graph is a cycle if it becomes a cycle after resetting all signs of all edge endpoints to $+1$. Apart from the trivial cycles, formed by the $1\times1$ incidence matrices $\pm2$ that correspond to self-loops, the incidence matrix of a cycle is a so-called \emph{hole matrix}~\cite{DBLP:journals/dm/ConfortiCV06}. Such matrix is of the form
\[
    \begin{pmatrix}
        \sigma_1&&&&&\tau_n\\
        \tau_1&\sigma_2&&&&\\
        &\tau_2&\sigma_3&&&\\
        &&&\ddots&&\\
        &&&&\sigma_{n-1}&\\
        &&&&\tau_{n-1}&\sigma_n
    \end{pmatrix}
\]
for $\sigma_1,\dots,\sigma_n,\tau_1,\dots,\tau_n\in\{-1,1\}$ after rearranging columns and rows. Laplace expansion along the first row shows that the determinant of the incidence matrix of a cycle is $\sigma_1\cdots\sigma_n+(-1)^{n+1}\tau_1\cdots\tau_n\in\{-2,0,2\}$. Since this determinant is nonzero if and only if $n$ is odd in simple graphs, we call a bidirected cycle \emph{odd} if the determinant of its incidence matrix is nonzero, i.e.\ the incidence matrix is an odd hole matrix.

Grossman et al.~\cite{GROSSMAN1995213} establish a well-known relation between the odd cycle packing number $\ocp G$ of a graph $G$ and the subdeterminants of its incidence matrix. Their proof, which shows that the absolute value of a subdeterminant is bounded by $2^{\ocp G}$, also applies to bidirected graphs when the definition odd cycle packing number $\ocp G$ of $G$ is straightforwardly extended using the earlier notion of odd cycles in bidirected graphs. That is, $\ocp G$ is the maximum cardinality of a collection of vertex disjoint odd cycles in $G$.

Another key insight from Grossman et al.~\cite{GROSSMAN1995213} is that a connected simple graph $G$ for which $|E|=|V|$ contains exactly a single cycle and that the absolute value of the determinant of the incidence matrix of $G$ is the absolute value of the determinant of the incidence matrix of this cycle. This can be obtained by successively Laplace expanding along rows with only one coefficient with an absolute value of $1$, which corresponds to a vertex of degree $1$. Note that after disregarding the sign of edges and disregarding the presence of half-edges, connectivity in bidirected graphs can be interpreted in terms of connectivity in the corresponding simple graphs, which is the viewpoint we will employ in the proof of \cref{lemma:circuit-ub-mixed}. In this context, the observation from~\cite{GROSSMAN1995213} straightforwardly generalizes to bidirected graphs as long as no half-edges are present in $G$. If a half-edge is present, the bidirected graph must consist of a tree with one additional half-edge. Following the same strategy of expanding along vertices incident to exactly one edge and not incident to the half-edge, such graph has an incidence matrix with determinant $\pm1$. In general, we may conclude that the determinant of a square incidence matrix of a connected bidirected graph $G$ is one of $\{0,\pm2^{\ocp G}$\}.

We are now ready to prove \cref{lemma:circuit-ub-mixed}.

\begin{proof}[Proof of \cref{lemma:circuit-ub-mixed}]
    Let $c\in\Z^{p+n}$ be a circuit of $A$ and let $B\in\Z^{(s-1)\times s}$ be a full row rank submatrix of $A$ with columns corresponding to the support of $c$ such that the projection $c'$ of $c$ to its support satisfies $Bc'=0$. By Cramer's rule, we have that $c'$ is a multiple of $\tilde c$ defined by
    \[
        \tilde c_j=(-1)^j\det\begin{bmatrix}B_{\cdot,1}&\cdots&B_{\cdot,j-1}&B_{\cdot,j+1}&\cdots&B_{\cdot,s}\end{bmatrix}\in\Z.
    \]
    Therefore, it suffices to show that the greatest common divisor of these determinants contains a sufficiently large factor $C$ so that $|\tilde c_j/C|$ is bounded by $\O(\Delta r)^{p+h}$.

    \proofsubparagraph*{Decomposing $B$ into many square connected components.} We have that $B$ is of the form
    \[
        B=\begin{pmatrix}
            \hat C&\hat W\\
            \hat T&\hat M
        \end{pmatrix},
    \]
    where $\hat C,\hat W,\hat T,\hat M$ are submatrices of $C,W,T$ and $M$ respectively. Let $\hat p$ be the width of $\hat T$ and $\hat h$ be the height of $\hat W$. The submatrix $\hat M$ retains an interpretation in terms of a (bidirected) subgraph $G(\hat M)$ of the graph $G(M)$. Observe that a connected component in $G(\hat M)$ corresponds to a set of vertices, rows of $\hat M$, and set of edges, columns of $\hat M$. In this way, we may permute the columns and rows of $\hat M$ so that it has a block-diagonal structure with blocks $D_k\in\Z^{m_k\times n_k}$ corresponding to the connected components of $G(\hat M)$ and a potential $m_0\times n_0$ block $D_0$ of zeroes, which may be interpreted as a degenerate component. Therefore, we may write
    \[
        B=\begin{pmatrix}
            \hat C&\hat W_0&\hat W_1&\cdots&\hat W_\ell\\
            \hat T_0&D_0&&&\\
            \hat T_1&&D_1&&\\
            
            \vdots&&&\ddots&\\
            \hat T_\ell&&&&D_\ell
        \end{pmatrix}.
    \]
    
    \proofsubparagraph*{Bounding the deviations of $D_k$ from being a square matrix.} We now show that almost all diagonal blocks $D_k$ have square dimensions. The step of arguing about the shape of the blocks $D_k$ resembles the argument used in~\cite{DBLP:conf/aaai/BrandKO21}. For this, note that the row space of $R_k:=\begin{pmatrix}\hat T_k&\matzero&D_k&\matzero\end{pmatrix}$, where $\matzero$ is a matrix of all zeroes, is of dimension $m_k$ as $B$ is full row rank. The (right-to-left) reduced row echelon form of the matrix reveals that this row space contains a linear subspace of $\R^{\hat p}\times\{\veczero\}$ of dimension $\max\{m_k-n_k,0\}$. For each $k$, let $S_k\subseteq\R^{\hat p}$ be a basis of this subspace with $|S_k|=\max\{m_k-n_k,0\}$. Suppose, in order to arrive at a contradiction, that the sum of the dimensions of these subspaces $\sum_{k\in[\ell]}|S_k|$ is greater than $\hat p$. As each subspace is a subspace of the $\hat p$ dimensional space $\R^{\hat p}\times\{\veczero\}$, there exists a nontrivial linear combination of vectors in the multiset $S:=\bigcup_{k\in[\ell]}S_k$ which yields the zero vector. Let $S'\subseteq S$ and $\lambda\in(\R\setminus\{0\})^{S'}=\veczero$ be so that $\sum_{s\in S'}\lambda_ss=\veczero$. Note that every $s\in S$ has an associated $k=k(s)$ such that $s\in S_k$ and that there exist associated multiplier vector $\mu^{(s)}\in\R^{m_k}$ such that $s=\sum_{i\in[m_k]}\mu_i^{(s)}(R_k)_{i,\cdot}^\top$. Expand the nontrivial linear combination yielding $\veczero$ to
    \[
        \veczero=\sum_{s\in S'}\lambda_s\sum_{i\in[m_{k(s)}]}\mu_i^{(s)}(R_{k(s)})_{i,\cdot}^\top=\sum_{k\in[\ell]}\sum_{i\in[m_k]}\left(\sum_{s\in S_k\cap S'}\lambda_s\mu_i^{(s)}\right)\cdot(R_k)_{i,\cdot}^\top
    \]
    As $\lambda_s\ne0$ and not all $\mu_i^{(s)}$ are $0$, we have that this yields a nontrivial linear combination of the rows of the matrix that is zero, contradicting the fact that $B$ has full row rank. Therefore, we may conclude that $\sum_{k\in[\ell]}\max\{m_k-n_k,0\}\le\hat p$. Analogously, by considering the full rank matrix $\binom{\tilde c^\top}B^\top$, we find that $\sum_{k\in[\ell]}\max\{n_k-m_k,0\}\le\hat h+1$. In particular, at least $\ell-\hat p-\hat h-1$ components $G(D_k)$ must be square in the sense that $n_k=m_k$.

    Note that for any $k$, the odd cycle packing number of $G(D_k)$ is bounded by
    \[
        \ocp G(D_k)\le n_k-(m_k-1)=n_k-m_k+1.
    \]
    This follows immediately from the fact that $G(D_k)$ consists of a spanning tree with $m_k-1$ edges and every additional edge can close at most one vertex disjoint cycle. Combining this with the bound on $n_k-m_k$ shows that the odd cycle packing number of the components is bounded.
    
    \proofsubparagraph*{The factorization of $\tilde c_j$ for a fixed $j$.} Consider an arbitrary $j\in[s]$. Let $\tilde B$ be the matrix $B$ after removing the $j$-th column. We estimate $|\tilde c_j|=|\det\tilde B|$, by successively Laplace expanding along the remaining columns of $\hat T$ and remaining parts of the rows of $\hat W$, of which there are at most $\hat p$ and $\hat h$ respectively. This shows that $\tilde c_j$ is the signed sum of at most $(s-1)^{\hat p}\cdot(s-\hat p)^{\hat h}\le r^{p+h}$ terms that are products of at most $\hat p+\hat h\le p+h$ coefficients of $A$ and a subdeterminant $\mathfrak m$ of a submatrix $\tilde M$ of $\hat M$. Such submatrix $\tilde M$ arises from $\hat M$ after deleting at most $\hat p+1$ columns and at most $\hat h$ rows. Note that $\ocp G(\tilde M)\le\ocp G(\hat M)$ as $G(\tilde M)$ may be interpreted as a subgraph of $G(\hat M)$, which shows that $\mathfrak m=0$ or $|\mathfrak m|=2^{\ocp G(\tilde M)}\le2^{\ocp G(\hat M)}$. Combining these findings, we obtain that $\tilde c_j$ is bounded by
    \[
        |\tilde c_j|\le r^{p+h}\cdot\Delta^{p+h}\cdot2^{\ocp G(\hat M)}=(\Delta r)^{p+h}\cdot2^{\ocp G(\hat M)}.
    \]
    
    We now inspect the divisibility of $\tilde c_j$ by a large power of $2$. Note that, because $G(\tilde M)$ arises from $G(\hat M)$ after removing at most $\hat p+\hat h+1$ vertices or edges, the graph $G(\tilde M)$ contains at least $(\ell-\hat p-\hat h-1)-(\hat p+\hat h+1)\ge\ell-2\hat p-2\hat h-2$ of the square components of $G(\hat M)$. Let $\tilde\ell=2\hat p+2\hat h+2$. Note that we may assume that $\ell-\tilde\ell\ge0$, because $\ell<2\hat p+2\hat h+2$ implies that
    \begin{align*}
        |\mathfrak m|\le2^{\ocp G(\tilde M)}&\le2^{\ocp G(\hat M)}\\
        &\le2^{\sum_{k\in[\ell]}(n_k-m_k+1)}\\
        &\le2^{\sum_{k\in[\ell]}\max\{n_k-m_k,0\}+\ell\cdot1}\\
        &\le2^{\hat h+1+\ell}<2^{2\hat p+3\hat h+3}\le2^{2p+3h+3},
    \end{align*}
    immediately showing that all $\tilde c_j$ are sufficiently small without requiring the division of any nontrivial common factor. Let $G(D_{k_1}),\dots,G(D_{k_{\ell-\tilde\ell}})$ be such square components appearing in $G(\tilde M)$. Then, $\tilde M$ can be rearranged into the block-diagonal form
    \[
        \tilde M=\begin{pmatrix}
            Q&&&\\
            &D_{k_1}&&\\
            
            &&\ddots&\\
            &&&D_{k_{\ell-\tilde\ell}}
        \end{pmatrix}
    \]
    where $Q$ represents the remaining part of $\tilde M$ that is not part of the square components. This shows that $\mathfrak m$ is divisible by $\det D_{k_1}\cdots\det D_{k_{\ell-\tilde\ell}}$. If any of these factors are zero, $\mathfrak m=0$ is divisible by any integer. Otherwise, it holds that $|\det D_{k_l}|=2^{\ocp G(D_{k_l})}$ because $G(D_{k_l})$ is connected and $D_{k_l}$ is square. Therefore, $\mathfrak m$ is divisible by $2^{\ocp G(D_{k_1})+\dots+\ocp G(D_{k_{\ell-\tilde\ell}})}$. In order to obtain a suitable power of $2$ that divides all $\tilde c_j$, we uniformly bound the difference
    \[
        \ocp G(\hat M)-(\ocp G(D_{k_1}))+\dots+\ocp G(D_{k_{\ell-\tilde\ell}})),
    \]
    which is equal to $\sum_{k\in S}\ocp G(D_k)$ for some complementary index set $S\subseteq[\ell]$ of size $\tilde\ell$. Using our observations, we have that this sum is at most
    \begin{align*}
        \sum_{k\in S}(n_k-m_k+1)&\le\sum_{k\in S}\max\{n_k-m_k,0\}+\tilde\ell\cdot1\\
        &\le\sum_{k\in[\ell]}\max\{n_k-m_k,0\}+\tilde\ell\\
        &\le\hat h+1+\tilde\ell\\
        &=2\hat p+3\hat h+3.
    \end{align*}
    Therefore, $\tilde c_j$ is divisible by $C:=2^{\ocp G(\hat M)-(2\hat p+3\hat h+3)}$, which is independent of $j$. After division by this common factor, we obtain integral $\tilde c_j/C$ bounded by
    \[
        |\tilde c_j/C|\le(\Delta r)^{p+h}\cdot2^{2\hat p+3\hat h+3}\le(\Delta r)^{p+h}\cdot2^{2p+3h+3}=\O(\Delta r)^{p+h},
    \]
    which completes the circuit bound.
\end{proof}

\cref{thm:proximity-from-circuits} now immediately implies \cref{cor:proximity-ub}, which we complement with \cref{prop:proximity-lb}. It shows that the bound from \cref{cor:proximity-ub} is asymptotically tight for constant $p$ and $h$ and $m=\Theta(n)$. It also implies that the circuits and Graver basis elements of the corresponding instances are at least of size $\Omega((\Delta m)^{p+h})$ in the $\infty$-norm.

\propproximitylb*

\begin{proof}
    Let $p,h\in\Z_{\ge0}$ and $\Delta\in\Z_{\ge1}$. Let $k\in\Z_{\ge1}$ be an arbitrary number determining the scale of the constructed instance which will satisfy $m,n=\Theta(k)$. We will use constraint matrix coefficients of $-1$ and $2$ in order to improve readability. Note that these coefficients can be eliminated through the use of \cref{lemma:master-reduction} to yield a class of instances using only coefficients from $\{0,1,\Delta\}$.
    
    Before involving the $p$ additional variables, we first present a construction that uses the complicating constraints to create a factor $(\Delta k)^h$ blowup. We essentially replicate the lower bound example
    \[
        \begin{pmatrix}
            \Delta'&-1&&\\
            &\Delta'&&\\
            &&\ddots&-1&\\
            &&&\Delta'&-1
        \end{pmatrix}
    \]
    in~\cite{DBLP:conf/ipco/HunkenschroderKKLL24} for some large $\Delta'$ by using smaller coefficients of size $\Delta\le\Delta'=\Delta k$ in the matrix $W$. First, define the auxiliary $k\times k$ block $U$ by
    \[
        U=\begin{pmatrix}
            -1&&&&&\\
            1&-1&&&&\\
            &1&-1&&&\\
            &&1&\ddots&&\\
            &&&\ddots&\ddots&\\
            &&&&1&-1\\
        \end{pmatrix},
    \]
    which will ensure that the variables associated with this block have equal value.

    \begin{itemize}
        \item If $h\ge1$: set
        \[
            \left(\begin{array}{c}
                \tilde W\\
                \cdashline{1-1}[3pt/3pt]
                \noalign{\vskip 2pt}
                \tilde M_W
            \end{array}\right)
            =
            \left(\begin{array}{c c c c c c c c c}
                &&\Delta\vecone^\top&-1&&&&&\\
                &&&&\Delta\vecone^\top&-1&&&\\
                &&&&&&\ddots&\ddots&\\
                &-1&&&&&&&\Delta\vecone^\top\\
                \cdashline{1-9}[3pt/3pt]
                e_1&&U&&&&&&\\
                &&&e_1&U&&&&\\
                &&&&&e_1&\ddots&&\\
                &&&&&&\ddots&\ddots&\\
                &&&&&&&e_1&U\\
            \end{array}\right),
        \]
        where $e_1$ denotes the unit vector $(1,0,\dots,0)$ of appropriate dimension.
        \item Otherwise, if $h=0$: set $\tilde M_W=\begin{pmatrix}1&-1\end{pmatrix}$ and $\tilde W$ being the matrix with zero rows.
    \end{itemize}
    
    By construction, solutions to
    \[
        \left(\begin{array}{c}
            \tilde W\\
            \cdashline{1-1}[3pt/3pt]
            \noalign{\vskip 2pt}
            \tilde M_W
        \end{array}\right)x=\veczero
    \]
    satisfy $x_2=(\Delta k)^hx_1$ and for any given $x_1'$, a unique such solution $x$ with $x_1=x_1'$ exists. It holds that $\tilde W\in\{-1,0,1,\Delta\}^{h\times\Theta(k)},\tilde M_W\in\Z^{\Theta(k)\times\Theta(k)}$ and that $\|\tilde M_W\|_1\le2$ with $\|(\tilde M_W)_{\cdot,1}\|_1=1$.

    We now proceed to construct instances of \crefilp{ilp:mixed} with high proximity and distinguish cases based on whether $p=0$.

    \begin{itemize}
        \item Case $p=0$: construct the system
        \[
            \left(\begin{array}{c c c}
                \noalign{\vskip 2pt}
                &&\tilde W\\
                \cdashline{1-3}[3pt/3pt]
                \noalign{\vskip 2pt}
                &&\tilde M_W\\
                2I&I&\\
                &\vecone^\top&-e_1^\top
            \end{array}\right)
            \left(\begin{array}{c}
                \alpha\\
                \beta\\
                x
            \end{array}\right)
            =
            \left(\begin{array}{c}
                \noalign{\vskip 2pt}
                \veczero\\
                \cdashline{1-1}[3pt/3pt]
                \noalign{\vskip 2pt}
                \veczero\\
                \vecone\\
                \veczero
            \end{array}\right)
            ,\quad
            \left(\begin{array}{c}
                \alpha\\
                \beta\\
                x
            \end{array}\right)\ge\veczero
        \]
        A vertex fractional solution with $\alpha=\tfrac12\vecone,\beta=\veczero,x_1=0$ exists whereas the only integral solution must assign $\alpha=\veczero,\beta=\vecone,x_1=k$. The previously introduced construction enforces that $x_2=(\Delta k)^hx_1$. Therefore, these solutions have $x_2=0$ and $x_2=k\cdot(\Delta k)^h$ respectively.
        
        \item Case $p>0$: we first create a large proximity using the $p$ additional variables with the system given by
        \[
            \left(\begin{array}{c;{3pt/3pt}c}
                \tilde T&\tilde M_T
            \end{array}\right)
            \left(\begin{array}{c}
                y\\
                \cdashline{1-1}[3pt/3pt]
                w\\
                \cdashline{1-1}[0.5pt/1pt]
                \gamma\\
                \delta
            \end{array}\right)
            =
            \left(\begin{array}{c}
                \vecone\\
                0\\
                \cdashline{1-1}[0.5pt/1pt]
                \veczero
            \end{array}\right),\quad
            \left(\begin{array}{c}
                y\\
                \cdashline{1-1}[3pt/3pt]
                w\\
                \cdashline{1-1}[0.5pt/1pt]
                \gamma\\
                \delta
            \end{array}\right)\ge\veczero,
        \]
        where
        \begin{align*}
            &\left(\begin{array}{c;{3pt/3pt}c}
                \tilde T&\tilde M_T
            \end{array}\right)
            =\\
            &\left(\begin{array}{c c c c c c;{3pt/3pt}c c c c c c c;{0.5pt/1pt}c c}
                &&&&&&&&&&&&&2I&I\\
                -e_1&&&&&&&&&&&&&&\vecone^\top\\
                \cdashline{1-15}[0.5pt/1pt]
                \Delta\vecone&&&&&&&I&&&&&&&\\
                &1&&&&&&\vecone^\top&&&&&&&\\
                &\Delta\vecone&&&&&&&I&&&&&&\\
                &&1&&&&&&\vecone^\top&&&&&&\\
                &&\Delta\vecone&&&&&&&I&&&&&\\
                &&&\ddots&&&&&&\vecone^\top&&&&&\\
                &&&\ddots&&&&&&&\ddots&&&&\\
                &&&&1&&&&&&\ddots&&&&\\
                &&&&\Delta\vecone&&&&&&&I&&&\\
                &&&&&1&&&&&&\vecone^\top&&&\\
                &&&&&\Delta\vecone&&&&&&&I&&\\
                &&&&&&1&&&&&&\vecone^\top&&
            \end{array}\right)
        \end{align*}
        and the identity matrices $I$ have dimensions $k\times k$. A solution to this system satisfies $w_1=(\Delta k)^py_1$ and $y_1=\vecone^\top\delta$. Again, there exists a vertex solution where $\delta=\veczero,w_1=0$ and the only integral solution has $\delta=\vecone,w_1=k\cdot(\Delta k)^p$. It holds that $\tilde T\in\{-1,0,1,\Delta\}^{\Theta(k)\times p},\tilde M_T\in\Z^{\Theta(k)\times\Theta(k)}$ and that $\|\tilde M_T\|_1\le2$ with $\|(\tilde M_T)_{\cdot,1}\|_1=1$.

        Now we can join this with the system with $h$ additional constraints along $u=(w,\gamma,\delta)$ to obtain the system
        \[
            \left(\begin{array}{c;{3pt/3pt}c c}
                &&\tilde W\\
                \cdashline{1-3}[3pt/3pt]
                \noalign{\vskip 2pt}
                &&\tilde M_W\\
                &e_1^\top&-e_1^\top\\
                \noalign{\vskip 2pt}
                \tilde T&\tilde M_T&\\
            \end{array}\right)
            \left(\begin{array}{c}
                y\\
                \cdashline{1-1}[3pt/3pt]
                u\\
                x
            \end{array}\right)
            =
            \left(\begin{array}{c}
                \veczero\\
                \cdashline{1-1}[3pt/3pt]
                \noalign{\vskip 2pt}
                \veczero\\
                0\\
                \noalign{\vskip 2pt}
                \tilde b\\
            \end{array}\right)
            ,\quad
            \left(\begin{array}{c}
                y\\
                \cdashline{1-1}[3pt/3pt]
                u\\
                x
            \end{array}\right)\ge\veczero.
        \]
        Here, a vertex fractional solution with $x_2=0$ exists, whereas the only integral solution has that $x_2=(\Delta k)^hx_1=(\Delta k)^hu_1=(\Delta k)^h\cdot k\cdot(\Delta k)^p=k\cdot(\Delta k)^{p+h}$.
    \end{itemize}

    It is straightforward to verify that for both constructions it holds that $m,n=\Theta(k)$. Therefore, we obtain the asymptotic lower bound of $\Omega(m\cdot(\Delta m)^{p+h})$.
    
\end{proof}

We note that a conjecture posed by Berndt, Mnich and Stamm involving upper bounds of the proximity of ILPs and their relation to the circuit complexity~\cite{DBLP:conf/sofsem/BerndtMS24} hints at that the upper bound from \cref{cor:proximity-ub} may be able to be improved to $m\cdot\O(\Delta m)^{p+h}$, at least for the case where $u=\vecinfty$, which would close the only remaining gap in terms of $m$ and $n$ between \cref{cor:proximity-ub,prop:proximity-lb}.

\subsection{A randomized slice-wise polynomial time algorithm}
\label{sec:mixed-algorithm}

The essence of solving \crefilp{ilp:mixed} in XP time when $A$ and $c$ are encoded in unary boils down to the following sequential steps:
\begin{enumerate}
    \item Reduce variable domains and eliminate the $p$ arbitrary variables $y$ in \cref{lemma:proximity-usage} using \cref{cor:proximity-ub}.
    \item Reduce to perfect constrained matching in \cref{lemma:wide-ip-reduction-to-perfect-matching} using the reduction to simple $G(M)$ in \cref{lemma:master-reduction} and the pseudo-polynomial reduction described in \cref{prop:gb}.
    \item Encode the $h$ constraints as a single constraint in \cref{lemma:constraint-condensation}.
    \item Solve the single-constrained perfect matching problem with the randomized algorithm by Mulmuley, Vazirani and Vazirani~\cite{DBLP:journals/combinatorica/MulmuleyVV87}, which is described in \cref{lemma:weighted-exact-matching}.
\end{enumerate}

To facilitate the first step, we rely on efficient computation of the solution to the LP relaxation of \crefilp{ilp:mixed} and \cref{cor:proximity-ub}. Solving the relaxation can be done in strongly polynomial time using the algorithm of Dadush et al.~\cite{DBLP:conf/focs/DadushNV20}. Their algorithm solves this relaxation in $(m+h)(n+p)^{\omega+1+o(1)}\log(\chi+n))$ time, where and $\omega$ is the matrix multiplication exponent $\chi(A)$ is the circuit imbalance measure. As $\chi(A)\le c_\infty(A)\le\O(\Delta(m+h))^{p+h}$ and $\omega\le3$, the LP relaxation can be solved in time $\tilde\O(mn^4)\cdot p^{\O(1)}h^{\O(1)}$. Here, the $\tilde\O$ hides polylogarithmic factors of the encoding length of the instance.

\begin{lemma}
    The \crefilp{ilp:mixed} can be solved by solving $n^p\cdot\O(\Delta(m+h))^{\tfrac12p(p+1)+ph}$ subinstances of \crefilp{ilp:wide} for which $c'=c,W'=W,M'=M$ and $\|u-l\|_\infty=n\cdot\O(\Delta(m+h))^h$, and solving an asymptotically equal number of LP relaxations of the original ILP with a subset of fixed variables.
    \label{lemma:proximity-usage}
\end{lemma}

\begin{proof}
    We solve the LP relaxation of \crefilp{ilp:mixed}. If the LP relaxation is unbounded, it suffices to solve the ILP for $(a,c)=\veczero$, which reduces to the bounded case. Let $(y^*,x^*)$ be an optimal fractional solution. Using \cref{cor:proximity-ub}, we may impose bounds on $(y,x)$ around $(y^*,x^*)$ so that $\|(g,u)-(e,l)\|_\infty=n\cdot\O(\Delta(m+h))^{p+h}$ and retain the optimal objective value of the integer linear program. We guess the value of the first complicating variable $y_1$ out of $n\cdot\O(\Delta(m+h))^{p+h}$ options. This reduces the problem to solving that many instances of \crefilp{ilp:mixed} with $p$ reduced by $1$ and a modified right hand side $(d',b')=(d,b)-\binom{W_{\cdot,1}}{T_{\cdot,1}}y_1$. Each such subinstance may analogously be recursively reduced until $p=0$. As $p$ decreases, the bound from \cref{cor:proximity-ub} strengthens and ensures that it suffices to only consider $n\cdot\O(\Delta(m+h))^{(p+1-j)+h}$ possible options for $y_j$. Therefore, the total resulting number of leaves in the recursion tree, i.e.\ subinstances with $p=0$, is at most
    \begin{align*}
        &\prod_{j\in[p]}n\cdot\O(\Delta(m+h))^{(p+1-j)+h}\\
        &=n^p\cdot\O(\Delta(m+h))^{\tfrac12p(p+1)+ph}.
    \end{align*}
    As the recursion tree grows exponentially, the number of LP relaxations that need to be solved is asymptotically equal to the number of leaves. The leaves have a proximity of $n\cdot\O(\Delta(m+h))^h$, which results in the claimed variable bound.
\end{proof}

We now describe the pseudo-polynomial reduction to a constrained perfect matching problem.

\begin{lemma}
    Solving \crefilp{ilp:wide} can be reduced to solving another instance of \crefilp{ilp:wide} where $l'=\veczero,u'=\vecone,b'=\vecone,c'\ge\veczero,W'\in\Z_{\ge0}^{h\times n'},\|c'\|_\infty=\|c\|_\infty,\Delta'=\Delta$ and $m'=\O(n\|u-l\|_\infty),n'\le m'^2$, and $G(M)$ is simple. This reduction can be performed in output-linear time.
    \label{lemma:wide-ip-reduction-to-perfect-matching}
\end{lemma}

\begin{proof}
    We apply \cref{lemma:master-reduction}. As a result, we may assume that the right hand side of this instance is bounded by $\|b\|_1=\O(\|u-l\|_1)=\O(n\|u-l\|_\infty)$ and that $M$ is the incidence matrix of the simple graph $G(M)$.
    
    We now further reduce to the case where $b=\vecone$ by using the graph $G(M)_b$ from \cref{prop:gb}. To accomplish this, let $M'$ be the incidence matrix of $G(M)_b$ and construct an ILP of the form \labelcref{ilp:wide} with $b=\vecone$ and the restrictions from \cref{lemma:master-reduction}. For every copy $j'$ of an edge $j$ of $G(M)$, copy the corresponding column $W_{\cdot,j}$ and use it as the column for $j'$ in $W'$. Copy the objective coefficients in the same way. This builds the correspondence $x_j=x_j^{(1)}+\dots+x_j^{(k)}$, where $x_j^{(1)},\dots,x_j^{(k)}$ are the $k=b_{i_1}\cdot b_{i_2}$ copies of the variable $x_j$ that corresponds with the edge connecting vertices $i_1$ and $i_2$ in $G(M)$. These problems are equivalent by \cref{prop:gb}. The graph $G(M)_b$ has $m'=\|b\|_1=\O(n\|u-l\|_\infty)$ vertices and $n'\le m^{\prime2}$ edges.
\end{proof}

The $h$ constraints of an instance of \crefilp{ilp:wide} which arises from the application of \cref{lemma:wide-ip-reduction-to-perfect-matching} can be condensed into a single constraint. This can accomplished by representing the constraints in base $B$ for a sufficiently large $B$, as done for special cases in \cite{DBLP:journals/mp/BergerBGS11,DBLP:journals/jacm/PapadimitriouY82} (for $h=2$ and $W\in\{0,1\}^{h\times n}$ respectively).

\begin{lemma}
    An instance of \crefilp{ilp:wide} with
    \begin{itemize}
        \item $W\in\Z_{\ge0}^{h\times n}$,
        \item $G(M)$ is simple,
        \item $l=\veczero,u=\vecone,b=\vecone$,
    \end{itemize}
    can be reduced to an instance of \crefilp{ilp:wide} with $h'=1,W'\in\Z_{\ge0}^{1\times n}$ in input+output linear time by only increasing $\Delta'=\Delta^h\cdot\O(m)^{h-1}$ and preserving $c,M,l,u,b$.
    \label{lemma:constraint-condensation}
\end{lemma}

\begin{proof}
    We condense all constraints of $W$ into one. Note that $(m/2)\Delta$ is a trivial upper bound on $W_{i,\cdot}x$ for a perfect matching $x$. Therefore, we may assume that $0\le d_i\le(m/2)\Delta$ for all $i$ or the ILP is trivially infeasible. By working base $B:=(m/2)\Delta+1$, we observe that the $h$ constraints
    \[
        Wx=d
    \]
    are equivalent to the single constraint
    \[
        \Bigl(\sum_{i\in[h]}B^{i-1}W_{i,\cdot}\Bigr)x=\sum_{i\in[h]}B^{i-1}d_i.
    \]
\end{proof}

We can now employ the randomized exact matching algorithm by Mulmuley, Vazirani and Vazirani~\cite{DBLP:journals/combinatorica/MulmuleyVV87} to solve instances resulting from \cref{lemma:constraint-condensation}. For the sake of completeness, we present an algorithm in \cref{lemma:weighted-exact-matching}, which closely follows their ideas.

\begin{lemma}
    A constrained minimum cost perfect matching problem, i.e., an instance of \crefilp{ilp:wide} with
    \begin{itemize}
        \item $c\ge\veczero,W\in\Z_{\ge0}^{1\times n},h=1$,
        \item $G(M)$ is simple,
        \item $l=\veczero,u=\vecone,b=\vecone$,
    \end{itemize}
    can be solved with a randomized algorithm in 
    \begin{align*}
        \O(\|c\|_\infty\Delta nm^7\log(\Delta m)\log(\|c\|_\infty\Delta nm))
    \end{align*}
    time. In this case, an optimal solution is not computed.
    \label{lemma:weighted-exact-matching}
\end{lemma}

\begin{proof}
    Assign every edge $j$ a random weight $w_j\gets(nm+1)c_j+Z_j$ where $Z_j$ is sampled independently and uniformly from $\{1,2,\dots,2n\}$. By doing so, the Isolation Lemma~\cite{DBLP:journals/combinatorica/MulmuleyVV87} ensures that with probability at least $\tfrac12$, there exists a unique minimum $w$-weighted perfect matching satisfying the additional $Wx=d$ constraint, if such matching exists. To check for the existence of such unique minimum matching, we compute the Pfaffian of the $m\times m$ skew-symmetric Tutte matrix $D$ where $D_{i_1,i_2}=0$ if $i_1$ and $i_2$ are nonadjacent and $D_{i_1,i_2}=2^{w_j}\cdot X^{W_{1j}}$ if $j$ connects $i_1<i_2$, where $X$ is some indeterminate. It is known that if there is a unique minimum $w$-weight constrained perfect matching, the Pfaffian will contain a monomial $C\cdot X^{d_1}$ with a nonzero coefficient $C$~\cite{DBLP:journals/combinatorica/MulmuleyVV87}. In fact, the minimum $w$-weighted perfect matching satisfying the additional constraint has a $w$-weight of $\omega$ which is the least integer for which $2^\omega$ divides $C$. As a perfect matching contains exactly $m/2$ edges, we have that the sum over $Z_j$ where $j$ ranges over this perfect matching is at most $nm$. Therefore, the minimum $c$-cost of the perfect matching can be recovered by computing the quotient of $\omega$ by $nm+1$. If there is no constrained perfect matching, the coefficient $C$ will always be zero. Therefore, we can compute the optimal objective of the ILP by computing the coefficient of $X^{d_1}$.
    
    The Pfaffian of the $m\times m$ matrix can be computed in a division free way with $\O(m^4)$ ring operations~\cite{DBLP:conf/cocoon/MahajanSV99}. As the ring elements can be large, these ring operations are expensive. A close inspection of the algorithm~\cite{DBLP:conf/cocoon/MahajanSV99}, see~\cite{DBLP:conf/icalp/LassotaL022}, reveals that the coefficients of the monomials in the computation can be bounded by $\O(m^m(2^{(nm+1)\|c\|_\infty+2n})^m)=\O(2^{m\log m+2nm+(nm+1)m\|c\|_\infty})$ and thus have encoding length bounded by $\overline l=\O(\|c\|_\infty nm^2)$. As we are solely interested in the coefficient $C$ of the term $C\cdot X^{d_1}$, monomials containing a power of $X$ greater than $d_1$ may be discarded during the computation of the Pfaffian. By again bounding $d_1\le(m/2)\Delta$, this shows that the polynomials have degree at most $\overline d:=(m/2)\Delta$. We can apply an FFT based fast univariate polynomial multiplication algorithm (Corollary 8.27~\cite{DBLP:books/daglib/0031325} using~\cite{harvey2021integer}) on this to obtain a multiplication time of
    \begin{align*}
        &\O(\overline d(\log\overline d+\overline l)\cdot\log(\overline d(\log\overline d+\overline l)))\\
        &=\O(\overline l\cdot\overline d\log\overline d\cdot\log(\overline l\cdot\overline d\log\overline d))\\
        &=\O\bigl(\|c\|_\infty\Delta nm^3\log(\Delta m)\cdot\log\bigl(\|c\|_\infty\Delta nm^3\log(\Delta m)\bigr)\bigr)\\
        &=\O(\|c\|_\infty\Delta nm^3\log(\Delta m)\log(\|c\|_\infty\Delta nm)).
    \end{align*}
    Here, we have slightly loosely estimated $\overline d(\log\overline d+\overline l)\le\overline l\cdot\overline d\log\overline d$ in order to improve readability. Additions can be performed in less asymptotic time. Multiplying with the $\mathcal O(m^4)$ arithmetic operations that need to be performed yields the stated running time.
\end{proof}

We can now combine the findings from this section to derive the claimed randomized XP time algorithm.

\begin{proof}[Proof of \cref{thm:mixed-xp}]
    Reduce solving \crefilp{ilp:mixed} to solving $n^p\cdot\O(\Delta(m+h))^{\tfrac12p(p+1)+ph}$ instances of constrained perfect matching with \cref{lemma:proximity-usage,lemma:wide-ip-reduction-to-perfect-matching}. Then apply \cref{lemma:constraint-condensation}. We find that we need to solve instances of size
    \begin{align*}
        m'&=n^2\cdot\O(\Delta(m+h))^h,\\
        n'&\le m'^2,\\
        \Delta'&=\Delta^h\cdot\O(m')^{h-1}.\\
    \end{align*}
    Finally, apply \cref{lemma:weighted-exact-matching} to find a running time of
    \begin{align*}
        &\O(\|c\|_\infty\Delta'n'm'^7\log(\Delta'm')\log(\|c\|_\infty\Delta'n'm'))\\
        &=\Delta^h\cdot\O(m')^{h+8}\cdot\|c\|_\infty\log(\Delta^h\cdot\O(m')^h)\log(\|c\|_\infty\cdot\Delta^h\cdot\O(m')^{h+2})\\
        &=\Delta^h\cdot(n^2\cdot\O(\Delta(m+h))^{h})^{h+8}\cdot\tilde\O(\|c\|_\infty\log\|c\|_\infty\cdot\log^2\Delta)\cdot h^{\O(1)}\\
        &=\tilde\O(\|c\|_\infty\log\|c\|_\infty\cdot\Delta^h\log^2\Delta)\cdot n^{2h+16}\cdot\O(\Delta(m+h))^{h^2+8h}
    \end{align*}
    for finding the optimal objective value of a single instance of \crefilp{ilp:wide}. Observe that polynomial factor in $h$ is subsumed by an exponentially growing factor $\O(1)^h$. To correctly find the subinstance of \crefilp{ilp:wide} that attains the maximal objective, each of the $n^p\cdot\O(\Delta(m+h))^{\tfrac12p(p+1)+ph}$ calls of the randomized algorithm must be successful. To compensate for this and retain a constant success probability of the overall algorithm, it suffices to attempt to solve each subinstance $\log(n^p\cdot\O(\Delta(m+h))^{\tfrac12p(p+1)+ph})=\tilde\O(\log\Delta)\cdot p^{\O(1)}h^{\O(1)}$ times with the randomized algorithm from \cref{lemma:weighted-exact-matching}. Therefore, we can obtain the optimal objective value of \crefilp{ilp:mixed} in time
    \begin{align*}
        &\tilde\O(\log\Delta)\cdot p^{\O(1)}h^{\O(1)}\cdot n^p\cdot\O(\Delta(m+h))^{\tfrac12p(p+1)+ph}\\
        &\,\,\,\,\,\,\,\,\cdot\tilde\O(\|c\|_\infty\log\|c\|_\infty\cdot\Delta^h\log^2\Delta)\cdot n^{2h+16}\cdot\O(\Delta(m+h))^{h^2+8h}\\
        &=\tilde\O(\|c\|_\infty\log\|c\|_\infty\cdot\Delta^h\log^3\Delta)\cdot n^{p+2h+16}\cdot\O(\Delta(m+h))^{\tfrac12p^2+h^2+ph+\tfrac12p+8h}.
    \end{align*}

    After identifying the optimal objective value and corresponding value of $y$, we can complete this to an optimal solution $(y,x)$ by performing a binary search on the value of each variable of $x$. In particular, we can recover the value of a variable $x_j$ in an optimal solution by restricting the domain of $x_j$ to $[\overline l_j,\underline u_j]$ and seeing if this still yields an optimal solution with the same objective value. In this way, for each $j\in[n]$ only logarithmically many bound pairs $\overline l_j,\underline u_j$ need to be tried to fix the values of $x_j$ and recover a solution. Using the imposed bounds on $\|u_j-l_j\|_\infty=n\cdot\O(\Delta(m+h))^h$, this requires solving the subinstance with additional variable restrictions an additional $\tilde\O(n\log\Delta)\cdot h^{\O(1)}$ times. Compensating for randomness yields an additional $\tilde\O(1)$ factor. The time to perform this last step and the time needed to solve the LP relaxations needed in \cref{lemma:proximity-usage} are both dominated by the time to compute the optimal objective value of the \crefilp{ilp:wide} instances.
\end{proof}

A Graver augmentation algorithm can be used to slightly improve the dependency on $n$. This minor modification is discussed in \cref{sec:graver-augmentation}.

Observe that, by using \cref{lemma:proximity-usage,lemma:wide-ip-reduction-to-perfect-matching} and a polynomial time algorithm for minimum cost perfect matching, one can rederive the fact that the generalized matching problem is in P, \cref{thm:generalized-matching-in-p}. Most of the steps taken in this reduction process from \cref{lemma:master-reduction} reduce to the steps as listed by Schrijver~\cite{schrijver2003combinatorial}. The difference lies in that the application of a sensitivity result and minimum cost circulation algorithm are replaced with a proximity based argument and an linear programming algorithm.

Instead of directly solving the constrained perfect matching problem via \cref{lemma:weighted-exact-matching}, we may further reduce the problem to the $0/1$-weighted exact matching problem, which is to find a perfect matching with a target weight, assuming all edges have weight $0$ or $1$. This can be done by splitting the unary-encoded coefficients in $W$ into multiple $0/1$ coefficients as shown in \cref{lemma:wide-coefficient-reduction}. This generalizes the reduction used in~\cite{DBLP:journals/jacm/PapadimitriouY82}. As the lemma additionally reveals that the W[1]-hardness of solving \crefilp{ilp:wide} persists even when we restrict to coefficients bounded by a constant $\Delta$ in \cref{thm:wide-w1-hard}, we provide an explicit proof.

\begin{lemma}
    An instance of \crefilp{ilp:wide} with finite $l,u$ and $W\in\Z_{\ge0}^{h\times n}$ can be reduced to an instance of \crefilp{ilp:wide} in output-linear time where
    \begin{itemize}
        \item $W'\in\{0,1\}^{h'\times n'}$,
        \item $\|l'\|_1=\O(\Delta\|l_1\|),\|u'\|_1=\O(\Delta\|u\|_1)$,
        \item $n',m'=\O(\Delta n+m)$,
        \item $h'=h$,
        \item $d'=d$,
        \item $\|b'\|_1=\O(\|b\|_1+\Delta\|l+u\|_1)$,
        \item $c'$ contains the same entries as $c$ up to the insertion of zeros.
    \end{itemize}
    
    In addition, the following properties of an instance are preserved:
    \begin{itemize}
        \item $G(M)$ is simple,
        \item $l=\veczero,u=\vecone,b=\vecone$.
    \end{itemize}
    \label{lemma:wide-coefficient-reduction}
\end{lemma}

\begin{proof}
    We describe a reduction step to reduce the coefficients of $W$ that are larger than $1$ in a given column $j$ by one. For this purpose, subdivide $x_j$ into three variables $x_j^{(1)},x_j^{(s)},x_j^{(2)}$. In essence, the coefficients of $W$ are reduced by splitting and distributing them over $x_j^{(1)}$ and $x_j^{(2)}$ as shown in \cref{fig:wide-coefficient-reduction}.
    \begin{figure}[H]
        \begin{cdisplaymath}
            \begin{array}{c:l}
                c_j&\\
                \cdashline{1-1}{}
                W_{\cdot,j}&\\
                \cdashline{1-1}{}
                \veczero&\\
                M_{i_1,j}&=b_{i_1}\\
                \veczero&\\
                M_{i_2,j}&=b_{i_2}\\
                \veczero&\\
                \\
                \\
                \cdashline{1-1}{}
                [l_j,u_j]&
            \end{array}
            \to
            \begin{array}{c c c:l}
                c_j&0&0\\
                \cdashline{1-3}{}
                (W_{\cdot,j}-\vecone)_+&0&\min\{W_{\cdot,j},\vecone\}\\
                \cdashline{1-3}{}
                \veczero&\veczero&\veczero\\
                M_{i_1,j}&0&0&=b_{i_1}\\
                \veczero&\veczero&\veczero\\
                0&0&M_{i_2,j}&=b_{i_2}\\
                \veczero&\veczero&\veczero\\
                1&1&0&=l_j+u_j\\
                0&1&1&=l_j+u_j\\
                \cdashline{1-3}{}
                [l_j,u_j]&[l_j,u_j]&[l_j,u_j]&
            \end{array}
        \end{cdisplaymath}
        \caption{Coefficients of the complicating constraints can be reduced by subdividing the variables and distributing the coefficient contributions among the new copies.}
        \label{fig:wide-coefficient-reduction}
    \end{figure}
    Again, let $M_{i_1,j},M_{i_2,j}$ denote the potentially nonzero coefficients with absolute value at most $1$ in $M$ of $x_j$ (interpreting $i_1=i_2$ when there is a single coefficient with absolute value $2$). We implement the following modifications:
    \begin{itemize}
        \item Add the constraints $x_j^{(1)}+x_j^{(s)}=l_j+u_j,x^{(2)}+x_j^{(s)}=l_j+u_j$ and domains $x_j^{(1)},x_j^{(s)},x_j^{(2)}\in[l_j,u_j]$.
        \item Replace the terms $M_{i_1,j}x_j$ and $M_{i_2,j}x_j$ with $M_{i_1,j}x_j^{(1)}$ and $M_{i_2,j}x_j^{(2)}$ respectively.
        \item Assign the objective coefficient $c_j$ to $x_j^{(1)}$ and assign all other new variables an objective coefficient of $0$.
        \item Set the column of $W'$ associated with $x_j^{(s)}$ to $\veczero$, but use the other two variables to reduce the coefficients in $W$. In particular, assign $(W_{\cdot,j}-\vecone)_+$ and $\min\{W_{\cdot,j},\vecone\}$ to $x_j^{(1)}$ and $x_j^{(2)}$ so that $W_{\cdot,j}x=W_{\cdot,j}x_j^{(1)}=(W_{\cdot,j}-\vecone)_+x_j^{(1)}+\min\{W_{\cdot,j},\vecone\}x_j^{(2)}$.
    \end{itemize}
    
    This reduces the coefficients in the $j$-th column by one. By exhaustively applying this reduction rule on all columns $j$, all coefficients of $W$ can be made $0/1$.
\end{proof}

We find that solving \crefilp{ilp:wide} is equivalent to the $0/1$-weighted exact matching problem.

\begin{proposition}
    If the objective coefficients and constraints are encoded in unary and $h$ is fixed, then solving \crefilp{ilp:wide} is polynomially equivalent to the $0/1$-weighted exact matching problem.
    \label{thm:wide-ip-exact-matching-equivalence}
\end{proposition}

\begin{proof}
    Reduce the problem to finding an optimal solution to a constrained perfect matching problem with \cref{lemma:proximity-usage,lemma:wide-ip-reduction-to-perfect-matching}. To solve \crefilp{ilp:wide}, perform binary search on the objective. The objective requirement can be encoded as an additional constraint and results in feasibility subproblems with $c=\veczero$. Then apply \cref{lemma:constraint-condensation,lemma:wide-coefficient-reduction} successively to find a polynomially equivalent $0/1$-weighted exact matching instance.
\end{proof}

In this way, a polynomial time deterministic algorithm for $0/1$-weighted exact matching yields a deterministic algorithm to solve \crefilp{ilp:wide}. The question of whether such algorithm exists is a long standing open question.

Additionally, if one is able to solve minimum cost perfect matching with one $0/1$ side constraint and costs encoded in binary, then \cref{lemma:proximity-usage,lemma:wide-ip-reduction-to-perfect-matching,lemma:constraint-condensation,lemma:wide-coefficient-reduction} show that \crefilp{ilp:wide} is similarly solvable. However, to the best of our knowledge, the complexity of this problem is unknown. In particular, it suffices to focus on minimum cost perfect matching under an additional $0/1$ budget constraint $\sum_{j\in[n]}w_jx_j\le B$ for $w\in\{0,1\}^n$, as one can reduce to this problem by decreasing the costs associated with edges with $w_j=1$ by a sufficiently large number $(m/2)\|c\|_\infty+1$ so that the budget constraint is forced to be met with equality.

\subsection{W[1]-hardness of matching with additional constraints and bounded coefficients}
\label{sec:mixed-hardness}

We complement \cref{thm:mixed-xp} with the hardness result listed in \cref{thm:wide-w1-hard}. For this, we employ the known hardness of ILP with coefficients encoded in unary~\cite{DBLP:journals/ai/DvorakEGKO21}. This problem was shown to be strongly W[1]-hard parameterized by the number of constraints through a reduction from the multicolored clique problem~\cite{DBLP:journals/ai/DvorakEGKO21}. For completeness, we present their reduction in a slightly more general setting to reduce from the partitioned subgraph isomorphism problem and simultaneously make it explicit that the variables may be restricted to $\veczero\le x\le\vecone$. With this, we also immediately obtain a fine-grained complexity lower bound due to Marx~\cite{DBLP:journals/toc/Marx10}.

\thmwonemulticolorclique*

We reduce from the partitioned subgraph isomorphism problem, which, given two graphs $H=(W,F),G=(V,E)$, and a partition of $V$ into $|W|$ classes $V_w$ for $w\in W$, asks whether there is a graph homomorphism $\phi\colon W\to V$ such that $\phi(w)\in V_w$. This problem is W[1]-hard when parameterized by the size $k$ of $H$. Note that the W[1]-hardness and Marx' fine-grained lower bound holds when $H$ is restricted to have $|F|=|\Theta(W)|$. Therefore, we may assume that the size parameter is given by $k=|F|+|W|$.

\begin{proof}
    Let $H=(W,F),G=(V,E),(V_w)_{w\in W}$, be an instance of the partitioned subgraph isomorphism problem. We construct a Sidon sequence $(a_v)_{v\in V}$ in $\Z_{\ge0}$ of size $a_v=\O(|V|^2)$, see~\cite{DBLP:journals/ai/DvorakEGKO21}. This is a sequence for which every sum of two sequence elements is unique. We proceed to describe an equivalent ILP instance. For every vertex $w\in W$ and $v\in V_w$, we create a variable $x_{wv}$, indicating whether $\phi(w)=v$, and for every edge $f\in F$ and $e\in E$, we create a variable $x_{ef}$, indicating whether $f$ is mapped to $e$ by $\phi$. Now add the constraints
    \begin{alignat}{2}
        \sum_{v\in V_{w_1}}a_vx_{w_1v}+\sum_{v\in V_{w_2}}a_vx_{w_2v}&=\sum_{e=\{v_1,v_2\}\in E}(a_{v_1}+a_{v_2})x_{fe},\quad&f=\{w_1,w_2\}&\in F\label{eq:reduction-adjacency}\\
        \sum_{v\in V_w}x_{wv}&=1,&w&\in W\label{eq:reduction-choose-vertex}\\
        \sum_{e\in E}x_{fe}&=1.&f&\in F\label{eq:reduction-choose-edge}
    \end{alignat}
    Together with $x\ge\veczero$, (\ref{eq:reduction-choose-vertex}) ensures that precisely one vertex $v$ is chosen for which $\phi(w)=v$. Note that imposing the upper bounds $x\le\vecone$ does not change the feasible set of solutions to the ILP. Constraint (\ref{eq:reduction-adjacency}) ensures that there is an edge $e$ between $v_1$ and $v_2$ if these vertices are chosen to represent an adjacent pair of vertices in $H$.

    A homomorphism $\phi$ respecting the partition corresponds to a feasible ILP solution defined by $x_{w,\phi(w)}=1$ for $w\in W,x_{f,\{\phi(w_1),\phi(w_2)\}}=1$ for $f=\{w_1,w_2\}\in F$ and letting $x$ be zero elsewhere.

    Now let $x$ be a feasible ILP instance. We construct the corresponding homomorphism $\phi$ by mapping $w$ to the unique vertex $\phi(w)=v\in V_w$ for which $x_{wv}=1$. To see that an edge $f=\{w_1,w_2\}\in F$ is mapped to an edge in $G$, note that (\ref{eq:reduction-adjacency}) reads $a_{\phi(w_1)}+a_{\phi(w_2)}=a_{v_1}+a_{v_2}$ where $e=\{v_1,v_2\}$ is the unique edge $e\in E(G)$ for which $x_{ef}=1$. This edge is unique by~(\ref{eq:reduction-choose-edge}). By the uniqueness of sums, it follows that $\{\phi(w_1),\phi(w_2)\}=e\in E$, i.e.\ that $\phi(w_1)$ is adjacent to $\phi(w_2)$ in $G$ as required.
\end{proof}

\cref{lemma:wide-coefficient-reduction} allows to derive the W[1]-hardness of solving \crefilp{ilp:wide} parameterized by $h$ for the restricted case of a $0/1$-constrained perfect matching problem on a bipartite graph for which many width parameters are constant.

\thmwidewonehard*

\begin{proof}
    We reduce an ILP instance
    \[
        \{Wx=d:\veczero\le x\le\vecone,x\in\Z^n\}
    \]
    of \cref{thm:w1-multicolorclique} to one of \crefilp{ilp:wide} with bounded coefficients. Observe that it can immediately be cast into the form \labelcref{ilp:wide} by setting $c=\veczero,M=\matzero$ and $b=\veczero$. To arrive at the listed class of incidence matrix, we add a redundant quadrangle for every variable $x_j$ as shown in \cref{fig:quadrangle}.
    \begin{figure}[H]
        \begin{cdisplaymath}
            \begin{array}{c}
                0\\
                \cdashline{1-1}{}
                W_{\cdot,j}\\
                \cdashline{1-1}{}
                \veczero\\
                \\
                \\
                \\
                \\
                \cdashline{1-1}{}
                [0,1]
            \end{array}
            \to
            \begin{array}{c c c c:c l}
                0&0&0&0&\\
                \cdashline{1-4}{}
                W_{\cdot,j}&\veczero&\veczero&\veczero&\\
                \cdashline{1-4}{}
                \veczero&\veczero&\veczero&\veczero&&\\
                1&0&0&1&=1\\
                1&1&0&0&=1\\
                0&1&1&0&=1\\
                0&0&1&1&=1\\
                \cdashline{1-4}{}
                [0,1]&[0,1]&[0,1]&[0,1]&
            \end{array}
        \end{cdisplaymath}
        \caption{A redundant quadrangle can be added for each variable to ensure that $G(M)$ is simple.}
        \label{fig:quadrangle}
    \end{figure}
    That is, we add the variables $x_j^{(2)},x_j^{(3)},x_j^{(4)}$ with zero objective coefficients and zero columns in $W$. We set all variable domains to $[0,1]$. Finally, the variables are given suitable coefficients in $M$ by adding the redundant constraints $x_j+x_j^{(2)}=x_j^{(2)}+x_j^{(3)}=x_j^{(3)}+x_j^{(4)}=x_j+x_j^{(4)}=1$.
    
    We now employ \cref{lemma:wide-coefficient-reduction} to reduce the coefficients of $W$. Observe that the reduction steps subdivides an edge into three edges. Therefore, the obtained $M$ is the incidence matrix of a simple graph that is the disjoint union of even length cycles.
\end{proof}

This shows that, unlike $n$-fold ILP, the constrained matching problem is unlikely to become FPT when the coefficient size is bounded.

As the reductions from \cref{thm:w1-multicolorclique,thm:wide-w1-hard} reduce a partitioned subgraph isomorphism problem to an instance of \crefilp{ilp:wide} with $h=\Theta(|F|+|W|)=\Theta(k)$ constraints in $W$ and $n=|V|^{\O(1)}$ variables, Corollary 6.3 in~\cite{DBLP:journals/toc/Marx10} immediately provides the lower bound in \cref{cor:wide-fine-grained-lb}.

\begin{corollary}
    If \crefilp{ilp:wide} with $W\in\{0,1\}^{h\times n}$ can be solved in time $f(h)n^{o(h/\log h)}$ for some function $f$, then ETH fails.
    \label{cor:wide-fine-grained-lb}
\end{corollary}

\section{Concluding notes and open problems}
We study the complexity of ILP parameterized by the size of variable or constraint backdoors to the generalized matching problem, called $p$ and $h$ respectively. In particular, we show that solving ILPs is in FPT when parameterized by $p$ and when $h=0$. We study the convexity of degree sequences in the light of SBO jump M-convex functions to express the non-backdoor variables efficiently as polyhedral constraints. Additionally, we present a randomized XP time algorithm to solve ILPs for fixed $p+h$ when the objective and constraint matrix are encoded in unary. This algorithm employs a circuit-based proximity bound to reduce the problem to the exact matching problem. Finally, we match this latter result by showing that ILP is W[1]-hard parameterized by $h$ when $p=0$.

Aside from the long-standing open question of whether the exact matching problem is in P, a natural open question coming from \cref{thm:mixed-xp} is whether the exponential dependence on the encoding length of $c$ is necessary. To answer this question positively, it suffices to restrict to $0/1$-budgeted perfect matching with its objective in encoded binary.

Finally, we suspect that some intermediate steps used in obtaining \cref{thm:tall-fpt} can be refined. That is, we suspect that \cref{lemma:convexity-of-sbo-jump-m-convex-functions} holds for general jump M-convex functions. In addition, it may be possible to make a combinatorial separation algorithm for $P_{r,U}$ in a more direct fashion as done in~\cite{DBLP:journals/mor/Zhang03}.

\bibliography{bib}

\appendix

\section{A Graver basis augmentation algorithm for ILP (W)}
\label{sec:graver-augmentation}

Graver basis augmentation methods have been an effective method to solve block-structured integer programming~\cite{eisenbrand2022algorithmictheoryintegerprogramming} and turn out to also be amendable to solve \crefilp{ilp:wide}.

In this part of the appendix, we show that by employing a Graver basis augmentation algorithm to solve \crefilp{ilp:wide}, the running time dependence on $n$ in
\[
    \tilde\O(\|c\|_\infty\log\|c\|_\infty\cdot\Delta^h\log^3\Delta)\cdot n^{p+2h+16}\cdot\O(\Delta(m+h))^{\tfrac12p^2+h^2+ph+\tfrac12p+8h}
\]
of the algorithm described in the proof of \cref{thm:mixed-xp} can be slightly improved. For this purpose, we prove \cref{prop:wide-xp} later in this section.

\begin{proposition}
   The optimal objective value of \crefilp{ilp:wide} can be computed in time $\tilde\O(\|c\|_\infty\log^2\|c\|_\infty\cdot\Delta^h\log^5\Delta)\cdot n^{h+10}\cdot\O(\Delta hm)^{h^2+8h}$ with a randomized algorithm.
    \label{prop:wide-xp}
\end{proposition}

By plugging \cref{prop:wide-xp} in the proof of \cref{thm:mixed-xp}, we immediately obtain \cref{thm:mixed-xp-alternative-running-time}.

\begin{corollary}
    \crefilp{ilp:mixed} can be solved in time
    \[
        \tilde\O(\|c\|_\infty\log^2\|c\|_\infty\cdot\Delta^h\log^6\Delta)\cdot n^{p+h+10}\cdot\O(\Delta hm)^{h^2+8h}\cdot\O(\Delta(m+h))^{\tfrac12p(p+1)+ph}.
    \]
    with a randomized algorithm.
    \label{thm:mixed-xp-alternative-running-time}
\end{corollary}

\begin{proof}
    We need to solve $n^p\cdot\O(\Delta(m+h))^{\tfrac12p(p+1)+ph}$ subinstances, each needing $\tilde\O(\|c\|_\infty\log^2\|c\|_\infty\cdot\Delta^h\log^5\Delta)\cdot n^{h+10}\cdot\O(\Delta hm)^{h^2+8h}$ time and solving each subinstance $\tilde\O(\log\Delta)\cdot p^{\O(1)}h^{\O(1)}$ times to compensate for randomness. This yields a total running time of
    \begin{align*}
        &\tilde\O(\log\Delta)\cdot p^{\O(1)}h^{\O(1)}\cdot n^p\cdot\O(\Delta(m+h))^{\tfrac12p(p+1)+ph}\\
        &\cdot\tilde\O(\|c\|_\infty\log^2\|c\|_\infty\cdot\Delta^h\log^5\Delta)\cdot\O(n)^{h+10}\cdot\O(\Delta hm)^{h^2+8h}\\
        &=\tilde\O(\|c\|_\infty\log^2\|c\|_\infty\cdot\Delta^h\log^6\Delta)\cdot\O(n)^{p+h+10}\cdot\O(\Delta hm)^{h^2+8h}\cdot\O(\Delta(m+h))^{\tfrac12p(p+1)+ph}
    \end{align*}
    to compute the optimal value of \crefilp{ilp:mixed}. Again, this dominates the time needed to compute relaxations and identify the solution attaining this value.
\end{proof}

Before proving \cref{prop:wide-xp}, we first introduce the Graver basis augmentation framework and relevant definitions. Intuitively, the Graver basis augmentation framework breaks solving an ILP down into multiple steps, each solving similar ``smaller'' ILPs, where the size depends on the complexity of the Graver basis.

\begin{definition}
    The Graver basis $\G(A)\subseteq\Z^n\setminus\{\veczero\}$ of an integer matrix $A\in\Z^{m\times n}$ is the set of conformally minimal nonzero integral kernel elements of $A$. That is, $g\in\G(A)$ if and only if $g'\in\Z^n\setminus\{\veczero\},Ag=\veczero$ and there is no $g'\in\Z^n\setminus\{\veczero,g\}$ such that $Ag'=\veczero$ and $g'\sqsubseteq g$.\\
    \label{def:graver-basis}
\end{definition}

The Graver complexity of $A$ refers to norm bounds on $\G(A)$. Our discussion will require bounds on $g_\infty(A)=\max\{\|g\|_\infty\ \vert\ g\in\G(A)\}$ and $g_1(A)=\max\{\|g\|_1\ \vert\ g\in\G(A)\}$, referring to the maximum $\infty$-norm and $1$-norm respectively. An oracle that can compute a single small augmenting step is called a Graver-best oracle.

\begin{definition}
    A Graver-best oracle is an oracle that, given a constraint matrix $A$, objective $c$, variable bounds $l$ and $u$, finds a Graver-best step. That is, it finds an integral solution $x$ to $\{Ax=\veczero,l\le x\le u\}$ such that
    \[
        c^\top x\le\min\left\{c^\top g\bigm\vert l\le g\le u,g\in\G(A)\right\}.
    \]
    \label{def:graver-best-oracle}
\end{definition}

Such an oracle can be implemented by solving an ILP over variable domains of size at most $\O(g_\infty(A))$, which we show to be polynomially bounded for fixed $h$ for \crefilp{ilp:wide}. Following the steps in \cref{sec:mixed-algorithm}, \cref{lemma:wide-ip-reduction-to-perfect-matching,lemma:constraint-condensation,lemma:weighted-exact-matching} show how to solve such \crefilp{ilp:wide}. \cref{thm:graver-augmentation} presents how a Graver-best oracle can be used to optimize an ILP.

\begin{theorem}[Lemma 12 in~\cite{eisenbrand2022algorithmictheoryintegerprogramming}, Lemma 5 in~\cite{DBLP:conf/icalp/EisenbrandHK18}]
    Given a feasible solution $x$ to \crefilp{ilp:general} with finite $l$ and $u$, an optimal solution to the ILP can be found with $\O(n\log(\|l-u\|_\infty)\log(c^\top x-c^\top x^*))$ queries of a Graver-best oracle. Here, $c^\top x^*$ is the value of an optimal solution.
    \label{thm:graver-augmentation}
\end{theorem}

By constructing an auxiliary ILP with a trivial feasible solution to an ILP of the form \labelcref{ilp:wide} as done in~\cite{DBLP:journals/mp/HemmeckeOR13}, we can use \cref{thm:graver-augmentation} to find a feasible solution as well as optimize such solution.

The polynomial bound on $g_\infty(A)$ for fixed $h$ follows from literature: Berndt, Mnich and Stamm~\cite{DBLP:conf/sofsem/BerndtMS24} show that the Graver basis elements of a matrix $M$ with $\|M\|_1\le2$ are small.

\begin{theorem}[Theorem 13 in~\cite{DBLP:conf/sofsem/BerndtMS24}]
    Let $M\in\Z^{m\times n}$ be an integer matrix with $\|M\|_1\le2$. Then $g_\infty(M)\le2$ and $g_1(M)\le2m+1$.
    \label{thm:graver-ub-generalized-matching}
\end{theorem}

Following the proof of Lemma 3 in~\cite{DBLP:conf/icalp/EisenbrandHK18}, which was originally used in the context of $n$-fold ILPs, and using the $1$-norm bound on $\G(M)$ from \cref{thm:graver-ub-generalized-matching}, and then measuring the Graver complexity in terms of the $\infty$-norm yields \cref{cor:graver-ub-wide}. For the sake of completeness, we provide an explicit proof.

\begin{corollary}
    Let $W\in\{-\Delta,\dots,\Delta\}^{h\times n}$ and $M\in\Z^{m\times n}$ satisfying $\|M\|_1\le2$. Then $g_\infty\bigl(\binom WM\bigr)\le2(2\Delta h(2m+1)+1)^h=\O(\Delta hm)^h$.
    \label{cor:graver-ub-wide}
\end{corollary}

\begin{proof}
    Let $g$ be a Graver basis element of $\binom WM$. In particular, $g\in\ker(M)$ and thus it can be decomposed into $g=g_1+\dots+g_\ell$ where all $g_k$ are sign-compatible Graver basis elements of $M$. These elements satisfy $\|g_k\|_\infty\le2$ and $\|g_k\|_1\le2m+1$ by \cref{thm:graver-ub-generalized-matching}. Consider the sum $0=Wg=Wg_1+\dots+Wg_\ell$ and observe that each term is bounded by $\|Wg_k\|_\infty\le\Delta\|g_k\|_1\le\Delta(2m+1)$. The Steinitz lemma~\cite{grinberg1980value} shows that there is a way to rearrange the terms of the sum so that each prefix sum is bounded in the $\infty$-norm by $h\cdot\Delta(2m+1)$. W.l.o.g. $g_1,\dots,g_\ell$ is in such order. If $\ell>(2\Delta h(2m+1)+1)^h$, there must be two prefix sums up to indices $k_1<k_2$ that sum up to the same value by the pigeonhole principle, which is impossible as it would reveal the decomposition of $g$ as the sum of the sign-compatible kernel elements $(g_1+\dots+g_{k_1}+g_{k_2+1}+\dots+g_k)$ and $(g_{k_1+1}+\dots+g_{k_2})$. Therefore, $\ell\le(2\Delta h(2m+1)+1)^h$, which, combined with $\|g_k\|_\infty\le2$, yields the result.
\end{proof}

The number of Graver-best steps that need to be computed in \cref{thm:graver-augmentation} can be polynomially bounded by computing a solution to the LP relaxation of \crefilp{ilp:wide} and using \cref{cor:proximity-ub}. Alternatively, one may employ Theorem 3.14 in~\cite{DBLP:journals/mp/HemmeckeKW14}, which shows that the proximity of an ILP is bounded by $n$ times the obtained Graver bound from \cref{cor:graver-ub-wide}. \cref{lemma:augmentation-step-count} makes this explicit.

\begin{lemma}
    \crefilp{ilp:wide} can be solved by performing $\tilde\O(h^2n\log(\|c\|_\infty\Delta)\log(\Delta))$ many Graver-best oracle queries for instances of \crefilp{ilp:wide} with $n'=\O(n),\|c'\|_\infty=\O(\|c\|_\infty)$ and $W'$ having the same entries as $W$ up to the insertion of an identity matrix; computing the LP relaxation of an asymptotically equally sized instance of \crefilp{ilp:wide}; and performing a preprocessing step which takes $\O(h^2n)$ time.
    \label{lemma:augmentation-step-count}
\end{lemma}

\begin{proof}
    We solve the LP relaxation of \crefilp{ilp:wide} to obtain an optimal rational solution $y^*$. Using this and \cref{cor:proximity-ub}, we may impose bounds on $x$ around $y^*$ so that $\|u-l\|_\infty=n\cdot\O(\Delta(m+h))^h$. Then we translate the ILP with $x'=x-l$ so that $l'=\veczero\le u'=u-l$. Note that we may assume that $\|(d,b)\|_\infty\le\|A\|_\infty\|u\|_\infty\le\Delta n\cdot n\cdot\O(\Delta(m+h))^h=\Delta n^2\cdot\O(\Delta(m+h))^h$ or the ILP is infeasible.

    If an initial feasible solution is not known, we can find such solution by optimizing over an auxiliary instance of \crefilp{ilp:wide} with a given feasible solution such as done in~\cite{DBLP:journals/mp/HemmeckeOR13}. This can be accomplished by building an instance of \crefilp{ilp:wide} with
    \[
        \begin{pmatrix}
            W'\\
            M'
        \end{pmatrix}=\begin{pmatrix}
            W&I&\matzero\\
            M&\matzero&I
        \end{pmatrix}
    \]
    for $h+m$ new variables $s$ corresponding to the right hand sides $(d,b)$. The constraint matrix satisfies $\|M'\|_1=\max\{\|M\|_1,1\}\le2$, showing that we have indeed constructed an instance of \crefilp{ilp:wide}. The objective corresponding to the original matching variables is set to $\veczero$. The variable bounds and objective coefficient corresponding to a variable $s_j$ are set to $0\le s_j\le(d,b)_j$ and $1$ if $(d,b)_j\ge0$ and $(d,b)_j\le s_j\le 0$ and $-1$ otherwise. All other parameters are left the same as for the restricted and translated ILP. As the lower bounds for $x$ are $0$, it is clear that $(\veczero,d,b)$ is a feasible solution to this auxiliary ILP unless the original ILP was trivially infeasible. Additionally, one can find a feasible solution $x$ to the original \crefilp{ilp:wide} by projecting the first $n$ components of a solution to the auxiliary instance that has optimal value $0$. If such solution does not exist, the original ILP is infeasible. The auxiliary instance has size $\|u'-l'\|_\infty=\Delta n^2\cdot\O(\Delta(m+h))^h$, $n'=\O(h+n)$ and $\|c'\|_\infty=\O(\|c\|_\infty)$. Note that we can assume that $h\le n$ after removing redundant constraints through Gaussian elimination in $\O(h^2n)$ arithmetic operations.

    We now inspect the number of Graver-best steps that the augmentation framework needs to compute for both finding the initial solution and improving it with respect to $c$. The difference in objective $c^\top x-c^\top x^*$ is at most $\|c\|_1\|u-l\|_\infty=\|c\|_1\cdot\Delta n^2\cdot\O(\Delta(m+h))^h$. For this reason, the Graver basis augmentation framework needs to compute at most $\O(n\log(\Delta n^2\cdot(\Delta(m+h))^h)\log(\|c\|_1\Delta n^2\cdot(\Delta(m+h))^h))=\tilde\O(h^2n\log(\|c\|_\infty\Delta)\log(\Delta))$ many Graver-best steps to find an optimal ILP solution by \cref{thm:graver-augmentation}.
\end{proof}

We can now combine the presented ingredients to complete the Graver basis augmentation algorithm.

\begin{proof}[Proof of \cref{prop:wide-xp}]
    Using \cref{lemma:augmentation-step-count} we can reduce finding the optimal objective of \crefilp{ilp:wide} to computing $\tilde\O(h^2n\log(\|c\|_\infty\Delta)\log(\Delta))$ Graver-best steps. Such Graver-best step can be found by solving an asymptotically equally sized instance of \crefilp{ilp:wide} with domains bounded by $\|u-l\|_\infty\le g_\infty(A)=\O(\Delta hm)^h$. By applying \cref{lemma:wide-ip-reduction-to-perfect-matching,lemma:constraint-condensation}, we find that we need to solve constrained perfect matching problems of size
    \begin{align*}
        m'&=n\cdot\O(\Delta hm)^h,\\
        n'&\le m'^2,\\
        \Delta'&=\Delta^h\cdot\O(m')^{h-1},\\
        \|c'\|_\infty&=\O(\|c\|_\infty).
    \end{align*}
    Finally, apply \cref{lemma:weighted-exact-matching} to find a running time of
    \begin{align*}
        &\O(\|c\|_\infty\Delta'n'm'^7\log(\Delta'm')\log(\|c\|_\infty\Delta'n'm'))\\
        &=\Delta^h\cdot\O(m')^{h+8}\cdot\|c\|_\infty\log(\O(m')^h\cdot\Delta^h)\log(\|c\|_\infty\cdot\Delta^h\cdot\O(m')^{h+2})\\
        &=\Delta^h\cdot(n\cdot\O(\Delta hm)^h)^{h+8}\cdot\tilde\O(\|c\|_\infty\log\|c\|_\infty\cdot\log^2\Delta)\cdot h^{\O(1)}\\
        &=\tilde\O(\|c\|_\infty\log\|c\|_\infty\cdot\Delta^h\log^2\Delta)\cdot n^{h+8}\cdot\O(\Delta hm)^{h^2+8h}
    \end{align*}
    for finding the objective improvement of a single Graver-best step. Note that \cref{thm:graver-augmentation} requires the actual solutions of the Graver-best step subproblems, which can again be recovered using a binary search on the variables. This entails computing each step subinstance with additional variable restrictions an additional $\tilde\O(n\log\Delta)\cdot h^{\O(1)}$ times. Compensating for randomness incurs another $\tilde O(1)$ factor of running time overhead for each Graver-best step that we compute.
    
    To compute all of the Graver-best steps successfully, we obtain a running time of
    \begin{align*}
        &\tilde\O\bigl(h^2n\log(\|c\|_\infty\Delta)\log(\Delta)\cdot\log\bigl(h^2n\log(\|c\|_\infty\Delta)\log(\Delta))\bigr)\\
        &\,\,\,\,\,\,\,\,\,\cdot\tilde\O(n\log\Delta)\cdot h^{\O(1)}\cdot\tilde\O(\|c\|_\infty\log\|c\|_\infty\cdot\Delta^h\log^2\Delta)\cdot n^{h+8}\cdot\O(\Delta hm)^{h^2+8h}\\
        &=\tilde\O(\|c\|_\infty\log^2\|c\|_\infty\cdot\Delta^h\log^5\Delta)\cdot n^{h+10}\cdot\O(\Delta hm)^{h^2+8h}.\\
    \end{align*}
    for solving all the required Graver-best oracle instances. Again, note that the setup time needed from \cref{lemma:augmentation-step-count} is dominated by the rest of the algorithm.
\end{proof}

\section{A direct MILP approach to exploit variable backdoors}
\label{sec:milp-approach-for-tall-fpt}

Recent independent and concurrent research by Eisenbrand and Rothvoss~\cite{eisenbrand2025parameterizedlinearformulationinteger} shows that the integer hull of arbitrary polyhedra $Ax\le b$ can be expressed as a system that depends linearly on $b$ as long as $b$ has a fixed remainder modulo a large integer depending only on the dimensions of $A$ and the coefficient size. They use this to optimize two-stage stochastic integer programming with large coefficients in the constraint matrix corresponding to the variable backdoor in FPT time. To accomplish this, they replace the non-backdoor variables with real variables and replace the constraints on those variables with their parametric versions of the integer hull. Their strategy reveals an additional way in which \cref{thm:tall-fpt} can be derived, which we shortly discuss.

Consider the perfect $b$-matching problem with additional variables resulting from \cref{lemma:master-reduction}. Let $P(z)$ be the polyhedron $\{x\in\R^n:Mx=z,x\ge0\}$ and $P_I(z)$ its integer hull $\conv(P(z)\cap\Z^n)$. It is known that $P_I(z)$ is given by
\[
    P_I(z)=\{x\in\R^n:Mx=z,x\ge0,Q_zx\le q\},
\]
where $Q_zx\le q$ consists of the constraints $\sum_{j\in\delta(U)}x_j\ge1$ for all $U\subseteq[m]$ such that $\sum_{i\in U}z_i$ is odd. See Corollary 31.2a in~\cite{schrijver2003combinatorial}. Here $\delta(U)$ are the edges connecting $U$ and its complement. These constraints are identical for all $z\equiv r\pmod2$ for a fixed remainder $r$, which justifies labeling the constraints as $Q_rx\le q$.

After guessing the remainder $r'\equiv y\pmod2$, we may solve the $b$-matching variant of \crefilp{ilp:tall} by solving the problem
\[
    \min\{a^\top y+c^\top x\ \vert\ y\in2\Z^p+r',\veczero\le y\le g,x\in P_I(b-Ty)\},
\]
as after finding an optimal $y$, we can obtain vertex solutions $x$ that are integral. This is equivalent to
\[
    \min\{a^\top(2v+r')+c^\top x\ \vert\ v\in\Z,\veczero\le 2v+r'\le g,x\in P_I(b-T(2v+r'))\}.
\]
As $z=b-T(2v+r')\equiv r\pmod2$ for all $v$ for suitably chosen $r$, we may restrict to solving
\[
    \min\{a^\top(2v+r')+c^\top x\ \vert\ v\in\Z,\veczero\le 2v+r'\le g,Mx=b-T(2v+r'),x\ge0,Q_rx\le q\}.
\]

This mixed integer linear program can be solved in FPT time parameterized by $p$ as a result of Lenstra's algorithm~\cite{DBLP:journals/mor/Lenstra83} as long as we can optimize any linear function over its linear relaxation in polynomial time. The latter is possible by using the ellipsoid method~\cite{DBLP:books/sp/GLS1988} and Padberg and Rao's odd minimum cut algorithm~\cite{DBLP:journals/mor/PadbergR82}, which allows to efficiently separate the inequalities in $Q_rx\le q$.

\end{document}